\documentclass[english]{article}
\usepackage[T1]{fontenc}
\usepackage[utf8]{luainputenc}
\usepackage{geometry}
\usepackage{amsmath}
\usepackage{amsthm}
\usepackage{amssymb}
\usepackage{esint}
\usepackage{centernot}
\usepackage{url}

\newcommand{\sign}{\textup{sign}}

\newcommand{\UL}{\textup{U}}

\newcommand{\dist}{\textup{dist}}

\newcommand{\LL}{\operatorname{Ext}}

\newcommand{\Cov}{\textup{Cov}}

\newtheorem{counter}{}[section]
\newtheorem{theorem}[counter]{Theorem}
\newtheorem{definition}[counter]{Definition}
\newtheorem{lemma}[counter]{Lemma}
\newtheorem{claim}[counter]{Claim}

\def\orient{\operatorname{Orient}}

\renewcommand{\S}{\mathcal{S}}
\renewcommand{\P}{\mathbb{P}}

\makeatletter
\newcommand\xleftrightarrow[2][]{%
  \ext@arrow 9999{\longleftrightarrowfill@}{#1}{#2}}
\newcommand\longleftrightarrowfill@{%
  \arrowfill@\leftarrow\relbar\rightarrow}
\makeatother

\makeatletter

\ifx\proof\undefined
\newenvironment{proof}[1][\protect\proofname]{\par
\normalfont\topsep6\p@\@plus6\p@\relax
\trivlist
\itemindent\parindent
\item[\hskip\labelsep\scshape #1]\ignorespaces
}{%
\endtrivlist\@endpefalse
}
\providecommand{\proofname}{Proof}
\fi
\theoremstyle{plain}

\theoremstyle{plain}

\usepackage{scalerel}
\def\shrinkage{-2.4mu}
\def\vecsign#1{\rule[1.388\LMex]{\dimexpr#1-2.5pt}{.36\LMpt}%
  \kern-6.0\LMpt\mathchar"017E}
\def\dvecsign#1{\rule{0pt}{7\LMpt}\smash{\stackon[-1.989\LMpt]{%
  \SavedStyle\mkern-\shrinkage\vecsign{#1}}%
  {\rotatebox{180}{$\SavedStyle\mkern-\shrinkage\vecsign{#1}$}}}}
\def\dvec#1{\ThisStyle{\setbox0=\hbox{$\SavedStyle#1$}%
  \def\useanchorwidth{T}\stackon[-4.2\LMpt]{\SavedStyle#1}{\,\dvecsign{\wd0}}}}
\usepackage{stackengine,amsmath}
\stackMath
\usepackage{graphicx}

\makeatother

\usepackage{babel}
\providecommand{\propositionname}{Proposition}
\providecommand{\theoremname}{Theorem}

\title{Rarity of extremal edges in random surfaces and other theoretical applications of cluster algorithms}
\author{Omri Cohen-Alloro\thanks{Supported by ISF grant~861/15 and by ERC starting grant 678520 (LocalOrder). School of Mathematical Sciences, Tel Aviv
University, Tel Aviv 69978, Israel. Emails: omrialloro@gmail.com,
peledron@post.tau.ac.il} \and Ron Peled\footnotemark[1] }

\begin{document}
\maketitle
\begin{abstract}
Motivated by questions on the delocalization of random surfaces, we prove that random surfaces satisfying a Lipschitz constraint rarely
develop extremal gradients. Previous proofs of this fact relied on
reflection positivity and were thus limited to random surfaces defined on
highly symmetric graphs, whereas our argument applies to general graphs. Our proof makes use of a cluster algorithm and reflection transformation for random surfaces of the type introduced by Swendsen-Wang, Wolff and Evertz et al. We discuss the general framework for such cluster algorithms, reviewing several particular cases with emphasis on their use in obtaining theoretical results. Two additional applications are presented: A reflection principle for random surfaces and a proof that pair correlations in the spin $O(n)$ model have monotone densities, strengthening
Griffiths' first inequality for such correlations.
\end{abstract}

\section{Introduction }\label{sec:introduction}
Our purpose in this paper is two-fold. First, we consider random
surface models satisfying a Lipschitz constraint, that is, random surfaces
whose gradients are constrained to be at most 1. For such surfaces
we prove that extremal gradients (close to 1 in magnitude) are very
unlikely to occur on any given set of edges. This is established for
all Lipschitz random surface models whose interaction potential is
monotone. The question of controlling the extremal gradients of
random surfaces was explicitly asked in \cite[Section
6]{milos2015delocalization} where such control was a key ingredient
in proving that Lipschitz (and more general) random surfaces
delocalize in two dimensions. Such a control was achieved in
\cite{milos2015delocalization} via the use of reflection positivity
(through the chessboard estimate) and as such was limited to random
surfaces defined on a torus graph. In contrast, our result applies
to random surfaces defined on an arbitrary, bounded degree, graph. New delocalization results may be obtained as a consequence as briefly discussed in Section~\ref{sec:discussion_and_open_questions}. Our proof makes use of a cluster algorithm and reflection transformation  for random surfaces of the type introduced by Swendsen-Wang \cite{swendsen1987nonuniversal}, Wolff \cite{wolff1989collective} (see also Brower
and Tamayo \cite{brower1989embedded}) and Evertz, Hasenbusch, Lana, Marcu, Pinn
and Solomon \cite{evertz1991stochastic, hasenbusch1992cluster}.

Our second goal is to discuss cluster algorithms of the above type in some generality. Such cluster algorithms, commonly used in Monte-Carlo simulation of the models, rely on finding a discrete, Ising-type, symmetry in the spin space
of the corresponding model (unlike the symmetries used in the reflection-positivity method which are symmetries of the underlying graph on which the model is defined). In Section~\ref{sec:reflection_transformations_cluster_algorithms} we discuss their general framework, reviewing in detail the cases of the Potts model, random surfaces, spin $O(n)$ model, Sheffield's cluster-swapping method and reversible Markov chains. Our review emphasizes the use of the algorithms in obtaining theoretical results and we demonstrate such use in two additional applications whose proof via the algorithms is relatively straightforward: A reflection principle for random surfaces and a proof that pair correlations in the spin $O(n)$ model have monotone densities, strengthening Griffiths' first inequality \cite{Gri67, Gin70} for such correlations.

We are by no means the first to discuss theoretical applications of cluster algorithms. Many such results are known in the literature including a work
of Aizenman \cite{aizenman1994slow}, following Patrascioiu and
Seiler \cite{patrascioiu1992phase}, on decay of correlations in
Lipschitz spin $O(2)$ models, a work of Burton and Steif \cite[Section 2]{burton1995new} on
characterizing the translation-invariant Gibbs states of a certain
subshift of finite type, works of Chayes-Machta \cite{chayes1997graphical, chayes1998graphical}, Chayes \cite{chayes1998discontinuity} and Campbell-Chayes \cite{campbell1998isotropic} relating phase transitions of spin systems with percolative properties of the graphical representation defined by their cluster algorithm, Sheffield's cluster swapping algorithm
\cite[Chapter~8]{sheffield2005random} used in the characterization
of translation-invariant gradient Gibbs states of random surfaces
(see also van den Berg \cite{van1993uniqueness} for a related
swapping idea used to study uniqueness of Gibbs measures) and a recent
work of Armend{\'a}riz, Ferrari and Soprano-Loto
\cite{armendariz2015phase} on phase transition in the dilute clock
model. However, these works mostly make use of ad-hoc transformations suitable to the task at hand and we feel that further emphasis of the unifying framework may still be of interest.

\subsection{Random surfaces}\label{sec:random_surfaces}
We begin by introducing the random surface model. Let $G = (V, E)$
be a finite connected graph (all our graphs will be simple,
undirected and without self-loops or multiple edges) and
$V_0\subseteq V$ be a non-empty subset of the vertices. Let $U$ be a
\emph{potential}, defined to be a measurable function
$U:\mathbb{R}\to(-\infty,\infty]$ satisfying $U(x)<\infty$ on a set
of positive Lebesgue measure and $U(x)=U(-x)$ for all $x$. The
random surface model with potential $U$, normalized to be $0$ at the
subset $V_0$, is the probability measure $\mu_{U,G,V_0}$ on
functions $\varphi:V\to\mathbb{R}$ defined by
\begin{equation}\label{eq:random_surface_measure}
  d\mu_{U, G, V_0}(\varphi) := \frac{1}{Z_{U, G, V_0}} \exp\Bigg(-\sum_{\{v,w\}\in E}
  U(\varphi_v - \varphi_w)\Bigg) \prod_{v\in V_0}\delta_0(d\varphi_v)\prod_{v\in V\setminus V_0}
  d\varphi_v,
\end{equation}
where $d\varphi_v$ denotes the Lebesgue measure on $\varphi_v$,
$\delta_0$ is a Dirac delta measure at $0$ and
\begin{equation*}
  Z_{U, G, V_0}:=\int \exp\Bigg(-\sum_{\{v,w\}\in E}
  U(\varphi_v - \varphi_w)\Bigg) \prod_{v\in V_0}\delta_0(d\varphi_v)\prod_{v\in V\setminus V_0}
  d\varphi_v
\end{equation*}
which we shall assume satisfies
\begin{equation}\label{eq:partition_function_condition}
  Z_{U, G, V_0} < \infty
\end{equation}
for $\mu_{U, G, V_0}$ to be well-defined (the fact that
$Z_{U,G,V_0}>0$ follows from our definition of potential).

For our applications we restrict attention to monotone potentials,
when $U$ satisfies
\begin{equation}\label{eq:monotonicity_condition_surfaces}
  U(x)\leq  U(y), \quad 0\leq x\leq y.
\end{equation}
This assumption implies that the density of a surface increases when
its gradients are decreased (in absolute value). In addition, we
often consider finitely-supported potentials in the sense that
\begin{align}
U(x)=\infty, \quad x>1.\label{eq:Lipschitz_condition_surfaces}
\end{align}
This assumption implies that a random surface configuration
$\varphi$ sampled from $\mu_{U,G,V_0}$ is a \emph{Lipschitz
function}, almost surely, in the sense that
\begin{equation}\label{eq:Lipschitz_function}
  |\varphi_v-\varphi_w|\le 1\quad\text{for all adjacent $v,w$.}
\end{equation}
We note that assumption \eqref{eq:partition_function_condition},
that $\mu_{U, G, V_0}$ is well-defined, is a consequence of
\eqref{eq:monotonicity_condition_surfaces} and
\eqref{eq:Lipschitz_condition_surfaces}.

An important example of a random surface satisfying
\eqref{eq:monotonicity_condition_surfaces} and
\eqref{eq:Lipschitz_condition_surfaces} is the case that $U$ is
given by the hammock potential,
\begin{equation}\label{eq:hammock_potential}
  U_{\text{hammock}}(x) = \begin{cases}
    0&|x|\le 1,\\
    \infty&|x|>1.
  \end{cases}
\end{equation}
In this case, the random surface is sampled uniformly among all
Lipschitz functions normalized to be~$0$ on $V_0$.

\subsubsection{Extremal gradients}
Our main result deals with random Lipschitz functions in the sense
of \eqref{eq:Lipschitz_function}. How rare are extremal gradients in
such surfaces, edges $\{v,w\}\in E$ on which $|\varphi_v -
\varphi_w|\ge 1-\varepsilon$ for some small $\varepsilon$? In
\cite[Section~6]{milos2015delocalization} it was asked whether,
under mild assumptions on the potential $U$, such gradients are
exponentially suppressed (`controlled gradients property') in the
sense that for each $\delta>0$ there exists $\varepsilon>0$,
depending only on $\delta$ and $U$ (not on the graph $G$), such
that for any distinct edges $\{\{v_i, w_i\}\}_{1\le i\le k}
\subseteq E$,
\begin{equation*}
  \mathbb{P}(|\varphi_{v_i} - \varphi_{w_i}|\ge 1-\varepsilon\text{ for all
  $i$})\le \delta^k.
\end{equation*}
A similar formulation was given for random surfaces with more
general potentials. The controlled gradients property was
established in \cite{milos2015delocalization}, for rather general
potential functions, when the graph $G$ is a torus in $\mathbb{Z}^d$
(a box with periodic boundary conditions), using reflection
positivity (via the chessboard estimate). This property was a key
ingredient in showing that two-dimensional random surfaces
delocalize for a large class of potential functions including the
hammock potential \eqref{eq:hammock_potential}. The work
\cite{milos2015delocalization} continues the delocalization results
of Brascamp, Lieb and Lebowitz \cite[Section
V]{brascamp1975statistical} and extends the class of potentials
treated there, so it is interesting to note that the arguments of
\cite{brascamp1975statistical} relied on a related property
\cite[inequality (16)]{brascamp1975statistical}. Our main result
establishes the controlled gradients property for Lipschitz random
surfaces with monotone potentials on general, bounded degree,
graphs.

\begin{theorem}\label{main_thm}
Let $G=(V,E)$ be a finite connected graph with maximal degree $\Delta$,
let $V_0\subseteq V$ be non-empty, let $U$ be a potential satisfying
\eqref{eq:monotonicity_condition_surfaces} and
\eqref{eq:Lipschitz_condition_surfaces} and let $\varphi$ be
randomly sampled from $\mu_{U,G,V_{0}}$. Then for any
$0<\varepsilon\le \frac{1}{8}$, $k\in\mathbb{N}$ and distinct
$\left\{ v_{1},w_{1}\right\} ,...,\left\{ v_{k},w_{k}\right\} \in
E$,
\begin{equation}\label{eq:main_result}
\mathbb{P}\left(\left\{
\left|\varphi_{v_{i}}-\varphi_{w_{i}}\right|\ge
1-\varepsilon\,:\,1\leq i\leq k\right\}
\right)\le\left(C(\Delta)\delta(U,\varepsilon)\right)^\frac{k}{C(\Delta)}
\end{equation}
where
\begin{equation*}
  \delta(U, \varepsilon):=\varepsilon\cdot \exp\left(-U(1-\varepsilon)+U(0)+\Delta\left(U\left(\frac{3}{4}\right) - U(0)\right)\right),
\end{equation*}
and where $C(\Delta)$ depends only on $\Delta$.
\end{theorem}
To illustrate the result we note that when $U = U_{\text{hammock}}$
we have $\delta(U, \varepsilon) = \varepsilon$ and, in addition,
that if $G$ is a tree then the probability in \eqref{eq:main_result}
exactly equals $\varepsilon^k$.

We note that the dependence on the maximal degree $\Delta$ in
\eqref{eq:main_result} cannot be completely removed. Indeed, suppose
$G = K_{n,n}$ is a complete bipartite graph with partite classes
$V_1, V_2$. Take the boundary set $V_0 = \{v_0\}$ for some $v_0\in
V_1$, take $U = U_{\text{hammock}}$ and let $\varphi$ be randomly
sampled from $\mu_{U,G,V_{0}}$. It is straightforward to check that
for any $0<\varepsilon<1$,
\begin{equation*}
  \mathbb{P}\left(\left\{
\left|\varphi_{v_{i}}-\varphi_{w_{i}}\right|\ge
1-\varepsilon\,:\,v_i\in V_1, w_i\in V_2\right\} \right)\ge
\mathbb{P}\left(\varphi(V_1)\subseteq\left[0,\frac{\varepsilon}{2}\right],\,
\varphi(V_2)\subseteq
\left[1-\frac{\varepsilon}{2},1\right]\right)\ge \left(\frac{\varepsilon}{4}\right)^{2n-1},
\end{equation*}
with the exponent $2n-1$ significantly smaller than the amount $n^2$
of edges between $V_1$ and $V_2$ in $G$.

Notwithstanding the above, we point out that Theorem~\ref{main_thm}
remains true if $\Delta$ is replaced by the maximal degree over all
vertices other than the vertices of $V_0$. In fact,
inequality~\eqref{eq:main_result} holds for a given set of edges
$\left\{ v_{1},w_{1}\right\} ,...,\left\{ v_{k},w_{k}\right\}$ when
$\Delta$ is replaced by the maximal degree of vertices in $\left\{
v_{1},w_{1},...,v_{k},w_{k}\right\}\setminus V_0$, and this is the
phrasing that we shall establish in the proof. This fact will allow
us to work with the graph in which the set $V_0$ is contracted to a
single vertex.

In Section~\ref{sec:discussion_and_open_questions} we briefly discuss the consequences of Theorem~\ref{main_thm} to
the delocalization of random surfaces.

\subsubsection{Reflection principle for random surfaces}
Let us first remind the reflection principle for simple random walk,
a fundamental relation between the distributions of the maximum of
the walk and the value at its endpoint. Let $(X_j)$, $0\le j\le n$,
be a simple random walk. That is, $X_0=0$ and $\{X_j-X_{j-1}\}$,
$1\le j\le n$, are independent increments, each uniformly
distributed on $\{-1,1\}$. Then the reflection principle states that
for all integer $k,m$ satisfying $m\ge\max\{0,k\}$,
\begin{equation}\label{eq:reflection_principle}
  \mathbb{P}(\max\{X_j\,:\,0\le j\le n\}\ge m,\, X_n=k) = \mathbb{P}(X_n = 2m-k).
\end{equation}
The law of the maximum of the walk is obtained as an immediate
consequence,
\begin{equation}\label{eq:maximum_of_walk}
  \mathbb{P}(\max\{X_j\,:\,0\le j\le n\}\ge m) = 2\mathbb{P}(X_n \ge m) - \mathbb{P}(X_n = m) = \mathbb{P}(|X_n|\ge m) - \mathbb{P}(X_n = m).
\end{equation}
In the standard proof of the reflection principle one reflects the
final segment of the walk around height $m$ and observes that this
is a one-to-one transformation between the events in the two sides
of \eqref{eq:reflection_principle}. As our main tool in this work is
a reflection transformation for random surfaces, one may naturally
wonder whether it yields an analogue of
\eqref{eq:reflection_principle}. This turns out to be the case, as
we now proceed to describe. We mention that while our main interest
is in random surfaces, the result applies equally well to random
walks (having symmetric increments with monotone densities), as
these can be seen as random surfaces on a line segment graph, and
yields a bound similar to that obtained from Doob's maximal
inequality.

We first describe what replaces the maximum in
\eqref{eq:reflection_principle}. Let $G=(V,E)$ be a graph,
$V_0\subseteq V$ be non-empty and $\varphi:V\to \mathbb{R}$. Let us
write $\{V_0 \xleftrightarrow{\varphi< m} v\}$ for the event that
there exists a path $v_0, v_1, \ldots, v_k$ in $G$ such that $v_0\in
V_0$, $v_k = v$ and $\varphi_{v_i}< m$ for all $i$. We write $\{V_0
\centernot{\xleftrightarrow{\varphi< m}} v\}$ for the complementary
event, that the `height barrier' between $V_0$ and $v$ is at least
$m$, meaning that on any path from $V_0$ to $v$ there is some vertex
$w$ with $\varphi_w\ge m$. We similarly define $\{V_0
\xleftrightarrow{\varphi>m} v\}$, etc.

Observe that in the one-dimensional case, when $V = \{0,1,\ldots,
n\}$ with $E = \{\{i,i+1\}\,:\,0\le i<n\}$ and $V_0 = \{0\}$, we
have $\{V_0 \centernot{\xleftrightarrow{\varphi< m}} n\} =
\{\max\{\varphi_i\,:\,0\le i\le n\}\ge m\}$, so that our definition
generalizes that of the maximum in \eqref{eq:reflection_principle}.

\begin{theorem}\label{thm:reflection_principle}
Let $G=(V,E)$ be a finite connected graph, let $V_0\subseteq V$ be
non-empty, let $U$ be a potential satisfying the monotonicity
condition \eqref{eq:monotonicity_condition_surfaces} and the
assumption \eqref{eq:partition_function_condition} that
$\mu_{U,G,V_0}$ is well-defined. Let $\varphi$ be randomly sampled
from $\mu_{U,G,V_{0}}$. Then
\begin{equation}\label{eq:barrier_ineq}
  \frac{1}{2}\mathbb{P}(|\varphi_v|\ge m)\le \mathbb{P}(V_0
\centernot{\xleftrightarrow{\varphi< m}} v)\le
\mathbb{P}(|\varphi_v|\ge m)\quad \text{for all $v\in V,
  m\ge 0$}.
\end{equation}
If, additionally, $U$ satisfies the finite-support condition
\eqref{eq:Lipschitz_condition_surfaces} then
\begin{equation}\label{eq:Lipschitz_barrier_ineq}
  \mathbb{P}(V_0
\centernot{\xleftrightarrow{\varphi< m}} v)\ge
\mathbb{P}(|\varphi_v|\ge m) - \mathbb{P}(\varphi_v\in(m,m+1))\quad
\text{for all $v\in V,
  m\ge 0$}.
\end{equation}
\end{theorem}
The above theorem gives an analogue of \eqref{eq:maximum_of_walk}
for random surfaces and our proof proceeds by first establishing an
analogue of \eqref{eq:reflection_principle}, see Section~\ref{sec:sublevel_set_connectivity}. We
remark that the lower bound in \eqref{eq:barrier_ineq} is trivial,
as $\frac{1}{2}\mathbb{P}(|\varphi_v|\ge m) =
\mathbb{P}(\varphi_v\ge m)$ and $\{\varphi_v\ge m\}\subseteq \{V_0
\centernot{\xleftrightarrow{\varphi< m}} v\}$. The improved lower
bound \eqref{eq:Lipschitz_barrier_ineq} does not hold without
additional assumptions (such as
\eqref{eq:Lipschitz_condition_surfaces}) as one can check on the
example of the single-edge graph, $V = \{0,1\}$, $E = \{\{0,1\}\}$,
$V_0 = \{0\}$ and $v = 1$, taking, e.g., $U(x) = x^2$ and $m$ large.

A discussion of the relation of the above results to the study of
excursion-set percolation of random surfaces appears in
Section~\ref{sec:discussion_and_open_questions}.

\subsection{Spin systems}\label{sec:spin_systems}
Our final result concerns the monotonicity of densities in
spin $O(n)$ models and is closely related to an inequality of
Armend{\'a}riz, Ferrari and Soprano-Loto \cite[Lemma
2.4]{armendariz2015phase} and Soprano-Loto \cite[Section
1.3.6]{soprano2015transicion}.

Let $G = (V, E)$ be a finite graph and
$V_0\subseteq V$ be a (possibly empty) subset of its vertices. Let $U:[-1,1]\to(-\infty,\infty]$ be a measurable function. The
spin $O(n)$ model with integer $n\ge 1$, potential function $U$ and normalized to equal $e_1:=(1,0,\ldots,0)\in\mathbb{S}^{n-1}$ at the
subset $V_0$, is the probability measure $\mu_{U, G, n, V_0}$ on
functions $\varphi:V\to\mathbb{S}^{n-1}$ defined by
\begin{equation}\label{eq:spin_O_n_measure_boundary_conditions}
  d\mu_{U, G, n, V_0}(\varphi) := \frac{1}{Z_{U, G, n, V_0}} \exp\Bigg(-\sum_{\{v,w\}\in E}
  U(\langle \varphi_v, \varphi_w\rangle)\Bigg) \prod_{v\in V\setminus V_0} d\mu_{\mathbb{S}^{n-1}}(\varphi_v) \prod_{v\in V_0} d\delta_{e_1}(\varphi_v),
\end{equation}
where $\langle \cdot, \cdot\rangle$ denotes the standard inner
product in $\mathbb{R}^n$, $\mu_{\mathbb{S}^{n-1}}$ denotes the uniform measure on $\mathbb{S}^{n-1}$,
$\delta_{e_1}$ is a Dirac delta measure at $e_1$ and
\begin{equation*}
  Z_{U, G, n, V_0} := \int \exp\Bigg(-\sum_{\{v,w\}\in E}
  U(\langle \varphi_v, \varphi_w\rangle)\Bigg) \prod_{v\in V\setminus V_0} d\mu_{\mathbb{S}^{n-1}}(\varphi_v) \prod_{v\in V_0} d\delta_{e_1}(\varphi_v)
\end{equation*}
which we shall assume satisfies
\begin{equation}\label{eq:partition_function_condition_spin_systems}
  0 < Z_{U, G, n, V_0} < \infty
\end{equation}
for $\mu_{U, G, n, V_0}$ to be well-defined. The standard spin $O(n)$ model is obtained
as the special case where $U(r) = -\beta r$, with $\beta$
representing the inverse temperature. Special cases of the standard spin $O(n)$ model have names of their own: The case $n=1$ is the Ising model, the case $n=2$ is the XY model, or plane rotator model, and the case $n=3$ is the Heisenberg model.

Observe that when $\varphi$ is randomly sampled from $\mu_{U, G, n, V_0}$,
the distribution of $\varphi_v$ is absolutely continuous with respect to $\mu_{\mathbb{S}^{n-1}}$ for each $v\in V\setminus V_0$. Denote its density function by $d_v$, so that $d_v:\mathbb{S}^{n-1}\to[0,\infty)$. Note that $d_v$ is only defined up to a $\mu_{\mathbb{S}^{n-1}}$-null set and that, by symmetry, there is a version of $d_v$ in which $d_v(b)$ is a function of $\langle b, e_1\rangle$. The next theorem states that
monotonicity of the potential function implies monotonicity of the densities $d_v$.

\begin{theorem}\label{thm:monotonicity_of_spin_O_n_density}
  Let $G = (V,E)$ be a finite graph and $V_0\subseteq V$ be a (possibly empty) subset of its vertices. Let $n\ge 1$ be an integer. Suppose that
  $U:[-1,1]\to(-\infty,\infty]$ is non-increasing in the sense that
  \begin{equation}\label{eq:U_monotone_spins}
    U(r)\ge U(s)\quad\text{for $r\le s$}
  \end{equation}
  and that \eqref{eq:partition_function_condition_spin_systems} holds. Let $\varphi$ be randomly sampled from $\mu_{U, G, n, V_0}$. Then for any
  $v\in V\setminus V_0$, there exists a version of the density $d_v$ satisfying
  \begin{equation}\label{eq:monotonicity_of_density}
    d_v(b_1) \ge d_v(b_2)\quad\text{when $\langle b_1, e_1\rangle\ge \langle b_2, e_1\rangle$}.
  \end{equation}
\end{theorem}
We make a few remarks regarding the theorem: The conclusion of the
theorem implies that $\mathbb{E}(\langle \sigma_v,
e_1\rangle)\ge 0$ for all $v\in V$, as in Griffiths first inequality \cite{Gri67, Gin70}. However, we are not aware that the monotonicity of the density has been noted in
earlier works, even for the standard ferromagnetic spin $O(n)$ model
(when $U(r) = -\beta r$ with $\beta>0$) with $n\ge 2$.

The result need not hold without the monotonicity condition
\eqref{eq:U_monotone_spins}. Indeed, monotonicity is a necessary
condition when $G$ is the single-edge graph $V = \{0,1\}$, $E =
\{\{0,1\}\}$ with $x=0$, $y=1$.

An analogous result holds for random surface models of the type
\eqref{eq:random_surface_measure} as we now state. This result,
however, follows easily from convexity considerations as explained
in Section \ref{sec:monotonicity_of_pair_correlations}.
\begin{theorem}\label{thm:non-increasing_density}
  Let $G=(V,E)$ be a finite connected graph, let $V_0\subseteq V$ be
non-empty, let $U$ be a potential satisfying the monotonicity
condition \eqref{eq:monotonicity_condition_surfaces} and the
assumption \eqref{eq:partition_function_condition} that
$\mu_{U,G,V_0}$ is well-defined. Let $\varphi$ be randomly sampled
from $\mu_{U,G,V_{0}}$. Then for any $x\in V$,
  \begin{equation*}
    |\varphi_x|\text{ has a non-increasing
    density with respect to Lebesgue measure on $[0,\infty)$}.
  \end{equation*}
\end{theorem}
We mention that analogous results to Theorem~\ref{thm:monotonicity_of_spin_O_n_density} and Theorem~\ref{thm:non-increasing_density} remain valid for clock models (spins on equally spaced points of $\mathbb{S}^1$) and \emph{integer-valued} height functions, respectively, and that our proofs apply to this setup without change.

\section{Cluster
algorithms and reflection transformations}\label{sec:reflection_transformations_cluster_algorithms}

In this section we describe the cluster algorithms and reflection transformations on which
our results are based. As mentioned in the introduction, such ideas
are not new, starting with the pioneering works of Swendsen-Wang \cite{swendsen1987nonuniversal} and
Wolff \cite{wolff1989collective}, they have been developed by many authors, mostly in the
context of fast simulation algorithms (cluster algorithms) but also
in theoretical contexts; see \cite{kandel1991general, caracciolo1993wolff, sokal1997monte, chayes1997graphical, chayes1998graphical, soprano2015transicion} for surveys and some recent
results. Nevertheless, we
believe that there is still room for presenting the special case
that we rely upon in some generality, highlighting connections with
previous works, to raise further awareness to the general framework
and its potential theoretical use. The general description below is followed by specific examples.

Let $(S,\S)$ be a measurable space and let $G = (V,E)$ be a finite
graph. Let $\vec{E}$ be an arbitrary orientation of the edges, i.e., $\vec{E}$ consists of either $(v,w)$ or $(w,v)$, but not both, for each edge $\{v,w\}\in E$. For each vertex $v\in V$, let $\lambda_v$ be a (finite or
infinite) measure on $(S,\S)$ and for each directed edge $(v,w)\in \vec{E}$, let
$h_{(v,w)}:S\times S\to [0,\infty)$ be a measurable function.
The model is defined by the probability measure $\mu_{\lambda, h,
G}$ on configurations $\varphi:V\to S$ given by
\begin{equation}\label{eq:measure_for_reflections}
  d\mu_{\lambda, h, G}(\varphi) = \frac{1}{Z_{\lambda, h, G}}
  \prod_{(v,w)\in \vec{E}} h_{(v,w)}(\varphi_v, \varphi_w) \prod_{v\in
  V} d\lambda_v(\varphi_v),
\end{equation}
where
\begin{equation*}
  Z_{\lambda, h, G} = \int \prod_{(v,w)\in \vec{E}} h_{(v,w)}(\varphi_v, \varphi_w) \prod_{v\in
  V} d\lambda_v(\varphi_v)
\end{equation*}
and we make the assumption that
\begin{equation}\label{eq:finite_partition_function}
0<Z_{\lambda,h,G}<\infty.
\end{equation}
In many of our examples $h$ will be specified on the undirected edge set $E$ but the possibility to define it on directed edges gives the model extra flexibility.

The reflection transformation is based on a function $\tau:S\to S$
with the following properties: For some set $V_0\subseteq V$,
\begin{align}
    &\text{$\tau$ is an involution:}&&\text{$\tau(\tau(s)) = s$ for
    all $s\in S$},\label{eq:tau_involution}\\
    &\text{$\tau$ preserves $\lambda$ for $v\notin V_0$:}&&\text{$\lambda_v\circ\tau^{-1} =
    \lambda_v$ for all $v\in V\setminus V_0$},\label{eq:tau_preserves_lambda}\\
    &\text{$\tau$ preserves $h$:}&&\text{$h_{(v,w)}(\tau(a), \tau(b)) =
    h_{(v,w)}(a,b)$ for all $(v,w)\in \vec{E},\,\, a,b\in
    S$}.\label{eq:tau_preserves_h}
\end{align}
Here $V_0$ plays the role of the `boundary' of $G$ in the sense that
we think of $(\lambda_v)$, $v\in V_0$, which are possibly
concentrated on a single value, as prescribing boundary conditions
for the measure $\mu_{\lambda, h, G}$. We allow the possibility that
$V_0 = \emptyset$ corresponding to free boundary conditions (but in
any case we require \eqref{eq:finite_partition_function}). We call a
$\tau$ satisfying the above properties a \emph{reflection}.

The reflection $\tau$ identifies `embedded Ising spins' in the model in the following sense. Suppose $\varphi$ is randomly sampled from $\mu_{\lambda, h, G}$ and define $\psi_v:=\{\varphi_v, \tau(\varphi_v)\}$, $v\in V$. Then, conditioned on $\psi$, each $\varphi_v$ has at most two possible values, and the joint distribution of these new `binary spins' is that of an Ising model (with general coupling constants, not necessarily of one sign, and zero magnetic field).

Fix a reflection $\tau$. We define a joint probability distribution,
the \emph{$\tau$-Edwards-Sokal coupling}, on pairs $(\varphi,
\omega)$ where $\varphi:V\to S$ is a configuration and
$\omega:E\to\{0,1\}$ may be thought of as a set of edges, by the
following prescription: The marginal distribution of $\varphi$ is
$\mu_{\lambda, h, G}$.
\begin{equation}\label{eq:omega_given_varphi}
  \text{Given $\varphi$, the $(\omega_{\{v,w\}})$ are
independent and satisfy $\mathbb{P}(\omega_{\{v,w\}} =
1\,|\,\varphi) =
  p_{\{v,w\}}(\varphi),\quad\{v,w\}\in E$},
\end{equation}
where, letting $(v,w)\in\vec{E}$ be the directed version of $\{v,w\}$,
\begin{equation}\label{eq:p v w formula}
  p_{\{v,w\}}(\varphi) := \max\left(1 -
  \frac{h_{(v,w)}(\tau(\varphi_v),
  \varphi_w)}{h_{(v,w)}(\varphi_v, \varphi_w)}, 0\right),
\end{equation}
where we set $\frac{0}{0}:=1$, $\frac{t}{0}=\infty$ for $t>0$ and we
note that $h_{(v,w)}(\tau(\varphi_v), \varphi_w)$ equals $h_{(v,w)}(\varphi_v, \tau(\varphi_w))$ due to the
assumptions \eqref{eq:tau_involution} and
\eqref{eq:tau_preserves_h} so the latter expression can be used instead of the former in~\eqref{eq:p v w formula}. Unfortunately, the marginal distribution
of $\omega$ does not seem to have a simple formula in general. We
note for later use the following immediate property,
\begin{equation}
  \text{Conditioned on $\varphi$, $\omega_{\{v,w\}} = 0$ almost surely for each $(v,w)\in\vec{E}$ with $h_{(v,w)}(\tau(\varphi_v),
  \varphi_w)\ge h_{(v,w)}(\varphi_v, \varphi_w)$}.
\end{equation}
In particular, if $\tau(\varphi_v) = \varphi_v$ then
$\omega_{\{v,w\}} = 0$ for all edges $\{v,w\}$ incident to $v$. We
also observe that if $h_{(v,w)}$ takes values in $\{0,1\}$, as
in the case of the hammock potential (see \eqref{eq:hammock_potential}), then
$p_{\{v,w\}}(\varphi)$ also belongs to $\{0,1\}$, almost surely, so
that $\omega$ is a deterministic function of $\varphi$.

We proceed to describe the reflection transformation. Let $(\varphi,
\omega)\in S^V\times \{0,1\}^E$ and let $x\in V$. Write $x
\xleftrightarrow{\omega} v$ if there is a path from $x$ to $v$ with
$\omega_{\{u,u'\}} = 1$ for all edges $\{u,u'\}$ along the path.
Similarly write $x \xleftrightarrow{\omega} V_0$ if there is some
$v_0\in V_0$ with $x \xleftrightarrow{\omega} v_0$ and write $x
\centernot{\xleftrightarrow{\omega}} v$ or $x
\centernot{\xleftrightarrow{\omega}} V_0$ for the non-existence of
such paths. The reflected configuration $\varphi^{\omega, x}:V\to S$
is defined by:
\begin{equation}\label{eq:flipped_configuration}
  \text{If $x \xleftrightarrow{\omega} V_0$ then $\varphi^{\omega, x}:=\varphi$. Otherwise, $\varphi^{\omega, x}_v := \begin{cases}
    \tau(\varphi_v)& x
\xleftrightarrow{\omega} v\\
    \varphi_v& x
\centernot{\xleftrightarrow{\omega}} v
  \end{cases}.$}
\end{equation}
That is, $\varphi^{\omega, x}$ is formed by applying $\tau$ to all
vertices in the $\omega$-connected component of $x$, unless this
connected component intersects $V_0$ in which case $\varphi^{\omega,
x} = \varphi$. The following lemma shows that this transformation
preserves the distribution of the $\tau$-Edwards-Sokal coupling.
\begin{lemma}\label{lem:flip_preserves_ES_coupling}
  Let $(\varphi, \omega)$ be randomly sampled from the
  $\tau$-Edwards-Sokal coupling. For each $x\in V$,
  \begin{equation}\label{eq:flip_preserves_distribution}
    (\varphi^{\omega, x}, \omega)\text{ has the same distribution as
    } (\varphi, \omega).
  \end{equation}
\end{lemma}
Of course, the equality in distribution \eqref{eq:flip_preserves_distribution} implies also that $\varphi^{\omega, x}$ has the same distribution as $\varphi$, leading to a natural Markov chain on configurations. In the context of the spin $O(n)$ model (see also below), the fact that $\tau$ is applied to a \emph{single} connected component of $\omega$ in each update is one of the innovations introduced by Wolff in his pioneering work \cite{wolff1989collective}. We remark, however, that the equality in distribution \eqref{eq:flip_preserves_distribution} remains true when the vertex $x$ is chosen as a function of $\omega$. That is, for any function $x:\{0,1\}^E\to V$,
\begin{equation*}
  (\varphi^{\omega, x(\omega)}, \omega)\text{ has the same distribution as } (\varphi, \omega).
\end{equation*}
More generally, $(\varphi, \omega)$ has the same distribution as $(\psi, \omega)$ where $\omega$ determines whether $\psi=\varphi$ or $\psi=\varphi^{\omega,x(\omega)}$. These facts can be deduced in a straightforward manner from \eqref{eq:flip_preserves_distribution} itself. By composing several operations of this type one may define various other measure-preserving transformations. For instance, in a Swendsen-Wang-type update, one applies $\tau$ with
probability $1/2$ independently to each of the $\omega$-connected
components which do not intersect $V_0$. Another possibility, when
$V_0$ is a singleton, is to either apply $\tau$ to the
$\omega$-connected component of $x$, or to the complement of this
connected component, according to whether the component intersects
$V_0$. We use a variant of this latter choice in Section \ref{extremal_gradients_sec} below.

One typical use of the reflection transformation described above and
Lemma~\ref{lem:flip_preserves_ES_coupling} is to define a reversible
Markov chain (a `cluster algorithm') on the set of configurations
with stationary distribution $\mu_{\lambda, h, G}$. A step of this
chain starting at $\varphi$ is conducted by deciding on some $x\in
V$ (possibly randomly), then sampling $\omega:E\to\{0,1\}$ from the
conditional distribution \eqref{eq:omega_given_varphi} and finally
outputting $\varphi^{\omega, x}$. Such Markov chains sometimes mix
faster than the more traditional Glauber dynamics, especially near
critical points of the model, and are used in practice in
Monte-Carlo simulations of the model (e.g., for Ising and spin
$O(n)$ models following Swendsen-Wang
\cite{swendsen1987nonuniversal} and Wolff
\cite{wolff1989collective}. See \cite{ullrich2014swendsen},
\cite{guo2017random} for recent polynomial-time mixing bounds). Our
emphasis, however, will be on theoretical applications.

The reflection transformation is defined above on a finite graph. It is also natural to work on infinite graphs, with the configuration $\varphi$ sampled from a \emph{Gibbs measure}, which is specified on finite graphs by distributions of the form~\eqref{eq:measure_for_reflections}, and with the distribution of $\omega$ given $\varphi$ specified by~\eqref{eq:omega_given_varphi} and~\eqref{eq:p v w formula}. In this case it may be shown, using Lemma~\ref{lem:flip_preserves_ES_coupling}, that reflections of \emph{finite} $\omega$-connected components preserve the joint distribution of $(\varphi,\omega)$. This may fail, however, when reflecting \emph{infinite} $\omega$-connected components. Still, it may be shown that reflections of infinite $\omega$-connected components transform the distribution of~$\varphi$ to that of another Gibbs measure, with the same underlying specification, while the distribution of~$\omega$ given $\varphi$ continues to be specified by~\eqref{eq:omega_given_varphi} and~\eqref{eq:p v w formula} (in the context of the cluster swapping reflection applied to random surfaces, see below, this was shown by Sheffield~\cite{sheffield2005random}; see also~\cite[Lemma 4.2]{chandgotia2018delocalization}). We do not discuss this further here.

\begin{proof}[Proof of Lemma~\ref{lem:flip_preserves_ES_coupling}]
It is sufficient to prove that
\begin{equation}\label{eq:product_sets_probabilities}
  \P(\cap_{v\in V}\{\varphi^{\omega, x}_v\in A_v\}\cap\{\omega = \omega_0\}) = \P(\cap_{v\in V}\{\varphi_v\in A_v\}\cap\{\omega = \omega_0\}),
\end{equation}
for all choices of $A_v\in\S$ for $v\in V$ and $\omega_0:E\to\{0,1\}$. Fix such $(A_v)$ and $\omega_0$. For brevity, we introduce the notation
\[
f_{(v,w)}(\varphi):=h_{(v,w)}(\varphi_v,
\varphi_w)p_{\{v,w\}}(\varphi)^{\omega_0({\{v,w\}})}\left(1-p_{\{v,w\}}(\varphi)\right)^{1-\omega_0({\{v,w\}})},\quad (v,w)\in \vec{E},
\]
with $0^0:=1$. Our definitions yield the following formula for the right-hand side of \eqref{eq:product_sets_probabilities},
\begin{equation}\label{eq:right-hand_side_expression}
  \P(\cap_{v\in V}\{\varphi_v\in A_v\}\cap\{\omega = \omega_0\}) = \frac{1}{Z_{\lambda, h, G}}\int\prod_{v\in V} d\lambda_v(\varphi_v)\mathbf{1}_{\varphi_v\in A_v}\prod_{(v,w)\in \vec{E}} f_{(v,w)}(\varphi).
\end{equation}
We proceed to evaluate the left-hand side of \eqref{eq:product_sets_probabilities}. Define $V_{\omega_0,x}$ to be the set of vertices on which $\tau$ is applied in the definition of $\varphi^{\omega_0,x}$. Precisely,
\[
\text{If $x \xleftrightarrow{\omega} V_0$ then $V_{\omega_0,x}:=\emptyset$. Otherwise, $V_{\omega_0,x} := \{v\in V\,\colon\, x
\xleftrightarrow{\omega} v\}$}.
\]
With this definition,
\begin{equation*}
  \bigcap_{v\in V}\{\varphi^{\omega, x}_v\in A_v\}\cap\{\omega = \omega_0\} = \bigcap_{v\in V_{\omega_0,x}}\{\tau(\varphi_v)\in A_v\}\bigcap_{v\in V\setminus V_{\omega_0,x}}\{\varphi_v\in A_v\}\cap\{\omega = \omega_0\},
\end{equation*}
whence the left-hand side of \eqref{eq:product_sets_probabilities} satisfies
\begin{multline*}
  \P(\cap_{v\in V}\{\varphi^{\omega, x}_v\in A_v\}\cap\{\omega = \omega_0\})\\
  = \frac{1}{Z_{\lambda, h, G}}\int\prod_{v\in V} d\lambda_v(\varphi_v)\prod_{v\in V_{\omega_0,x}}\mathbf{1}_{\tau(\varphi_v)\in A_v}\prod_{v\in V\setminus V_{\omega_0,x}}\mathbf{1}_{\varphi_v\in A_v}\prod_{(v,w)\in \vec{E}} f_{(v,w)}(\varphi).
\end{multline*}
To simplify the last integral, we make the change of variables $\varphi\mapsto\psi$ where
\begin{equation*}
  \psi_v := \begin{cases}
    \tau(\varphi_v) & v\in V_{\omega_0, x}\\
    \varphi_v& v\notin V_{\omega_0, x}
  \end{cases}.
\end{equation*}
This mapping is one-to-one as $\tau$ is invertible by \eqref{eq:tau_involution}. In addition, it preserves the measure $\prod_{v\in V} \lambda_v$ by \eqref{eq:tau_preserves_lambda} and the fact that $V_0\cap V_{\omega_0, x}=\emptyset$. Thus,
\begin{equation}\label{eq:left-hand_side_expression}
  \P(\cap_{v\in V}\{\varphi^{\omega, x}_v\in A_v\}\cap\{\omega = \omega_0\}) = \frac{1}{Z_{\lambda, h, G}}\int\prod_{v\in V} d\lambda_v(\psi_v)\mathbf{1}_{\psi_v\in A_v}\prod_{(v,w)\in \vec{E}} f_{(v,w)}(\varphi).
\end{equation}
Comparing \eqref{eq:right-hand_side_expression} and \eqref{eq:left-hand_side_expression} we see that \eqref{eq:product_sets_probabilities} is a consequence of
\begin{equation}\label{eq:equality_of_f_factors}
  f_{(v,w)}(\varphi) = f_{(v,w)}(\psi)
\end{equation}
for all $(v,w)\in \vec{E}$. The equality \eqref{eq:equality_of_f_factors} is trivial in the case that $v,w\notin V_{\omega_0,x}$. In the case that $v,w\in V_{\omega_0,x}$ it follows from \eqref{eq:tau_preserves_h} by using that
\begin{align*}
  &h_{(v,w)}(\psi_v, \psi_w) = h_{(v,w)}(\tau(\varphi_v),
\tau(\varphi_w)) = h_{(v,w)}(\varphi_v, \varphi_w),\\
  &p_{\{v,w\}}(\psi) = \max\left(1 - \frac{h_{(v,w)}(\tau(\tau(\varphi_v)),
  \tau(\varphi_w))}{h_{(v,w)}(\tau(\varphi_v), \tau(\varphi_w))}, 0\right) = \max\left(1 -
  \frac{h_{(v,w)}(\tau(\varphi_v),
  \varphi_w)}{h_{(v,w)}(\varphi_v, \varphi_w)}, 0\right) = p_{\{v,w\}}(\varphi).
\end{align*}
Lastly, in the case that, say, $v\in V_{\omega_0, x}$ and $w\notin V_{\omega_0, x}$, using now the involution property \eqref{eq:tau_involution},
\begin{align*}
  f_{(v,w)}(\psi) &= h_{(v,w)}(\psi_v,\psi_w)(1-p_{\{v,w\}}(\psi)) = h_{(v,w)}(\tau(\varphi_v),\varphi_w)\min\left(\frac{h_{(v,w)}(\tau(\tau(\varphi_v)),
  \varphi_w)}{h_{(v,w)}(\tau(\varphi_v), \varphi_w)}, 1\right)\\
   &= h_{(v,w)}(\tau(\varphi_v),\varphi_w)\min\left(\frac{h_{(v,w)}(\varphi_v,
  \varphi_w)}{h_{(v,w)}(\tau(\varphi_v), \varphi_w)}, 1\right) \stackrel{(*)}{=} h_{(v,w)}(\varphi_v,\varphi_w)\min\left(\frac{h_{(v,w)}(\tau(\varphi_v),
  \varphi_w)}{h_{(v,w)}(\varphi_v, \varphi_w)}, 1\right)\\
   &= h_{(v,w)}(\varphi_v,\varphi_w)(1-p_{\{v,w\}}(\varphi)) = f_{(v,w)}(\varphi),
\end{align*}
where the equality $(*)$ follows by separately considering the two cases $h_{(v,w)}(\tau(\varphi_v),\varphi_w)\ge h_{(v,w)}(\varphi_v, \varphi_w)$ and $h_{(v,w)}(\tau(\varphi_v),\varphi_w)<h_{(v,w)}(\varphi_v, \varphi_w)$. This finishes the proof of the lemma.
\end{proof}

We now illustrate the general construction above with specific examples:

{\bf Potts model:} The $q$-state Potts model (with free boundary
conditions) is obtained by taking $S = \{1,2,\ldots, q\}$ (with the
discrete sigma algebra), all $\lambda_v$ equal to the counting
measure, $V_0=\emptyset$ and
\begin{equation*}
  h_{\{v,w\}}(a,b) = \exp(\beta \delta_{a,b}),\quad \{v,w\}\in E,\,\, a,b\in
  S,
\end{equation*}
where $\delta_{a,b}$ is the Kronecker delta function. Let $\tau:S\to
S$ be any involution. One checks simply that $\tau$ satisfies
\eqref{eq:tau_preserves_lambda} and \eqref{eq:tau_preserves_h} so
that $\tau$ is a reflection.

The model is called ferromagnetic when the parameter $\beta$ is non-negative. In this case, the general prescription \eqref{eq:omega_given_varphi} for the
distribution of $\omega$ becomes
\begin{equation*}
  \P(\omega_{\{v,w\}} = 1\, |\, \varphi) = \begin{cases}
    1 - \exp(-\beta) & \varphi_v = \varphi_w\text{ and }\tau(\varphi_v)\neq \varphi_v\\
    0 & \text{otherwise}
  \end{cases}.
\end{equation*}
In particular, the value of $\varphi$ is constant on each
$\omega$-connected component.
Applying $\tau$ to the values of $\varphi$ on $\omega$-connected
components leads to a variant of the Swendsen-Wang Markov chain (the original Markov chain is obtained by setting $\omega_{\{v,w\}}=1$ with probability $1-\exp(-\beta)$, independently, for each $\{v,w\}\in E$ with $\varphi_v = \varphi_w$ and $\omega_{\{v,w\}}=0$ for other edges. Then updating $\varphi$ by assigning it a uniform value in $S$ on each $\omega$-connected component, independently. In this case the marginal distribution of $\omega$ is explicit and given by the $q$-random cluster model).

The model is called anti-ferromagnetic when $\beta$ is negative. In this case, the general prescription \eqref{eq:omega_given_varphi} becomes
\begin{equation*}
  \P(\omega_{\{v,w\}} = 1\, |\, \varphi) = \begin{cases}
    1 - \exp(\beta) & \tau(\varphi_v) = \varphi_w\text{ and }\varphi_v\neq \varphi_w\\
    0 & \text{otherwise}
  \end{cases}.
\end{equation*}
In particular, on each path in the subgraph given by
$\omega^{-1}(1)$ the value of $\varphi$ alternates between two
distinct values $a,b\in S$ with $b = \tau(a)$. We note in passing
that in the limiting case $\beta = -\infty$, corresponding to
$\varphi$ being a uniformly sampled proper $q$-coloring, the
$\omega$-connected components are exactly the Kempe chains of the
pairs $a,b\in S$ with $b = \tau(a)$. In the anti-ferromagnetic case,
applying $\tau$ to the values of $\varphi$ on $\omega$-connected
components leads to the Wang-Swendsen-Koteck\'y \cite{wang1989antiferromagnetic} Markov
chain.

{\bf Random surfaces:} The random surface measure $\mu_{U,G,V_0}$
defined in \eqref{eq:random_surface_measure}, having potential $U$
and normalized to be $0$ on $V_0$, fits the framework
\eqref{eq:measure_for_reflections} by taking $S = \mathbb{R}$ (with
the Borel sigma algebra), $\lambda_v$ to be Lebesgue measure for
$v\notin V_0$, $\lambda_v = \delta_0$ for $v\in V_0$ and
$h_{\{v,w\}}(a,b) = \exp(-U(a-b))$ for all $\{v,w\}\in E$.

For each $m\in\mathbb{R}$ let $\tau_m:\mathbb{R}\to\mathbb{R}$ be
the `reflection around $m$' mapping. That is,
\begin{equation}\label{eq:random_surfaces_reflection}
\tau_m(a) = 2m - a.
\end{equation}
It is straightforward to check the conditions
\eqref{eq:tau_involution}, \eqref{eq:tau_preserves_lambda} and
\eqref{eq:tau_preserves_h}, using that $U(x) = U(-x)$ for all $x$, and conclude that $\tau_m$ is a
reflection for any random surface measure $\mu_{U,G,V_0}$. In this
case, the general prescription \eqref{eq:omega_given_varphi} for the
distribution of $\omega$ becomes
\begin{equation}\label{eq:random_surface_omega_probability}
  \P(\omega_{\{v,w\}} = 1\, |\, \varphi) = \max\left(1 -
  \exp(U(\varphi_v - \varphi_w) - U(2m - \varphi_v - \varphi_w)),\, 0\right).
\end{equation}
The following consequence plays a central role in the proofs of our
main theorems:
\begin{equation}\label{eq:random_surface_omega_connection_restriction}
\begin{split}
  &\text{If $U$ is monotone in the sense of
\eqref{eq:monotonicity_condition_surfaces} then, almost surely, on each $\omega$-connected component $C$}\\
  &\text{either
$\varphi_v>m$ for all $v\in C$, or $\varphi_v<m$ for all $v\in
C$, or $C = \{v\}$ and $\varphi_v = m$.}
\end{split}
\end{equation}
This follows from \eqref{eq:random_surface_omega_probability} by
noting that if $U$ is monotone and $\varphi_v\ge m\ge \varphi_w$
then $\varphi_v - \varphi_w \ge |2m-\varphi_v-\varphi_w|$ whence
$U(\varphi_v - \varphi_w) \ge U(2m - \varphi_v - \varphi_w)$.

Extensions of the above ideas to integer-valued random surfaces
(when a measure on $\varphi:S\to\mathbb{Z}$ is defined analogously)
follow in a similar manner, with the reflection height $m$
restricted to $\mathbb{Z}\cup(\mathbb{Z}+\frac{1}{2})$. The reflection principle for simple random walk (see \eqref{eq:reflection_principle}) can be seen as a reflection transformation of an integer-valued random surface on a one-dimensional graph (see also the discussion of reflection transformations for Markov chains below). Markov chain
algorithms for simulating random surfaces based on the above ideas
were developed by Evertz, Hasenbusch, Lana, Marcu, Pinn and Solomon
\cite{evertz1991stochastic, hasenbusch1992cluster}.

We are not aware of a simple expression for the marginal
distribution of $\omega$.

{\bf Spin $O(n)$ model:} The spin $O(n)$ measure $\mu_{U,G,n,V_0}$
defined in \eqref{eq:spin_O_n_measure_boundary_conditions}, having integer $n\ge 1$, potential $U$
and normalized to equal $e_1=(1,0,\ldots,0)$ on $V_0$, fits the framework \eqref{eq:measure_for_reflections} by taking $S =
\mathbb{S}^{n-1}\subseteq\mathbb{R}^n$ (with the sigma algebra
inherited from $\mathbb{R}^n$), $d\lambda_v = d\varphi_v$ for $v\in V\setminus V_0$, $\lambda_v = \delta_{e_1}$ for $v\in V_0$ and $h_{\{v,w\}}(a,b)
= \exp(-U(\langle a,b\rangle))$ for all $\{v,w\}\in E$.

For each $a\in\mathbb{S}^{n-1}$ let
$\tau_a:\mathbb{S}^{n-1}\to\mathbb{S}^{n-1}$ be the `reflection
around the hyperplane orthogonal to $a$' mapping. That is,
\begin{equation}\label{eq:spin_system_reflection}
\tau_a(b) = b - 2\langle a,b\rangle\cdot a.
\end{equation}
It is straightforward
to check the conditions \eqref{eq:tau_involution},
\eqref{eq:tau_preserves_lambda} and \eqref{eq:tau_preserves_h}, verifying along the way that $\tau_a$ is an isometry, and
conclude that $\tau_a$ is a reflection for any spin $O(n)$ measure
$\mu_{U,G,n,V_0}$. In this case, the general prescription
\eqref{eq:omega_given_varphi} for the distribution of $\omega$
becomes
\begin{equation}\label{eq:spin_On_omega_probability}
  \P(\omega_{\{v,w\}} = 1\, |\, \varphi) = \max\left(1 -
  \exp(U(\langle \varphi_v,\, \varphi_w\rangle) - U(\langle\varphi_v - 2\langle a,\varphi_v\rangle a,\, \varphi_w\rangle)),\, 0\right).
\end{equation}
Again, the following consequence plays a central role in our proof
of Theorem~\ref{thm:monotonicity_of_spin_O_n_density}:
\begin{equation}\label{spin_one_sided_reflection}
\begin{split}
  &\text{If $U$ is non-increasing in the sense of
\eqref{eq:U_monotone_spins} then, almost surely, on each $\omega$-connected component $C$}\\
  &\text{either
$\langle a,\varphi_v\rangle> 0$ for all $v\in C$, or $\langle
a,\varphi_v\rangle< 0$ for all $v\in C$, or $C = \{v\}$ and $\langle a,\varphi_v\rangle = 0$.}
\end{split}
\end{equation}
This follows from \eqref{eq:spin_On_omega_probability} by noting
that if $U$ is non-increasing and $\langle a,\varphi_v\rangle\langle
a,\varphi_w\rangle\le 0$ then $\langle \varphi_v,\, \varphi_w\rangle
\le \langle\varphi_v - 2\langle a,\varphi_v\rangle a,\,
\varphi_w\rangle$ whence $U(\langle \varphi_v,\, \varphi_w\rangle)
\ge U(\langle\varphi_v - 2\langle a,\varphi_v\rangle a,\,
\varphi_w\rangle)$.

Extensions of the above ideas to clock models (when $\varphi$ takes
values in a set of $q$ equally-spaced marks on $\mathbb{S}^1$)
follow in a similar manner, with the vector $a$ restricted to be one
of the marks or exactly in between two marks. Wolff's algorithm
\cite{wolff1989collective} pioneered the use of the above ideas to
fast simulation algorithms for the spin $O(n)$ model.

We are not aware of a simple expression for the marginal
distribution of $\omega$. Nevertheless, Chayes \cite{chayes1998discontinuity},
Chayes-Campbell \cite{campbell1998isotropic}, following Chayes-Machta \cite{chayes1997graphical, chayes1998graphical}, have
considered the standard spin $O(n)$ model with $n\in \{2,3\}$ and
proved that the distribution of $\omega$ has positive association
(every two monotone increasing functions of $\omega$ are
non-negatively correlated), that an infinite $\omega$-connected
component (in a suitable infinite-volume limit) arises if and only
if there is positive magnetization in the spin model and related
results.

{\bf Cluster swapping:} The term cluster
swapping was coined by Sheffield \cite[Chapter
8]{sheffield2005random} for the following setup, in the special
setting of random surfaces. A related swapping idea was used by van den Berg \cite{van1993uniqueness} to study uniqueness of Gibbs measures (see also van den Berg and Steif~\cite[Proof of Theorem 2.4]{van1994percolation}). Let $V_0\subseteq V$. Let $\mu_{\lambda^1, h, G}$, $\mu_{\lambda^2, h, G}$ be general measures of the type
\eqref{eq:measure_for_reflections}, with the same $h$ and with
$\lambda^1_v = \lambda^2_v$ for all $v\in V\setminus V_0$. Let
$\varphi^1, \varphi^2$ be independent samples from $\mu_{\lambda^1,
h, G}$ and $\mu_{\lambda^2, h, G}$ respectively. We regard the pair
$(\varphi^1, \varphi^2)$ as a configuration $(\varphi^1,
\varphi^2):V\to S\times S$ which is sampled from the measure
$\mu_{\lambda^1\times\lambda^2, h\times h, G}$, where
\begin{equation*}
  (\lambda^1\times\lambda^2)_v := \lambda^1_v \times \lambda^2_v,\quad v\in V
\end{equation*}
and
\begin{equation*}
  (h\times h)_{(v,w)}((a_1, a_2), (b_1, b_2)) := h_{(v,w)}(a_1,
  b_1)\cdot h_{(v,w)}(a_2, b_2),\quad (v,w)\in \vec{E},\,\,
  a_1,a_2,b_1,b_2\in S.
\end{equation*}
Let $\tau:S\times S\to S\times S$ be the `swap' mapping, defined by
\begin{equation}\label{eq:cluster_swapping_reflection}
\tau((a_1, a_2)) = (a_2, a_1).
\end{equation}
It is straightforward to check the
conditions \eqref{eq:tau_involution},
\eqref{eq:tau_preserves_lambda} and \eqref{eq:tau_preserves_h} and
conclude that $\tau$ is a reflection for
$\mu_{\lambda^1\times\lambda^2, h\times h, G}$. In this case, the
general prescription \eqref{eq:omega_given_varphi} for the
distribution of $\omega$ becomes
\begin{equation}\label{eq:cluster_swap_omega_probability}
  \P(\omega_{\{v,w\}} = 1\, |\, \varphi) = \max\left(1 -
  \frac{h_{(v,w)}(\varphi^2_v, \varphi^1_w)\cdot h_{(v,w)}(\varphi^1_v,\varphi^2_w)}{h_{(v,w)}(\varphi^1_v, \varphi^1_w)\cdot
  h_{(v,w)}(\varphi^2_v,\varphi^2_w)},\, 0\right).
\end{equation}
We see that, as remarked before, if $\varphi^1_v = \varphi^2_v$ then
$\omega_{\{v,w\}} = 0$ for all $w$ adjacent to $v$. Thus, an
$\omega$-connected component is `blocked' by places where the two
coordinates of $(\varphi^1, \varphi^2)$ are equal. This observation
is closely related to the theorem of van den Berg \cite[Theorem 1]{van1993uniqueness} showing the equality of two Gibbs measures under the assumption that
there is no `disagreement percolation' between independent samples
from the measures.

Sheffield \cite[Lemma 8.1.3]{sheffield2005random} made an additional
important observation in the context of random surfaces. The setup considered there allows for both integer-valued and real-valued surfaces, and also for surfaces whose potential $U$ is not required to satisfy the restriction $U(x)=U(-x)$ (allowing to introduce a slope to the surface). We explain Sheffield's observation in the real-valued case: Take $S = \mathbb{R}$, $\lambda^1_v, \lambda_v^2$ to be Lebesgue
measure for $v\in V\setminus V_0$ and $h_{(v,w)}(a,b) =
\exp(-U(a-b))$ for a measurable $U:\mathbb{R}\to(-\infty,\infty]$ and all $(v,w)\in \vec{E}$. Then,
\begin{equation*}
\begin{split}
  &\text{if $U$ is convex then, almost surely, on each $\omega$-connected component $C$}\\
  &\text{either $\varphi^1_v> \varphi^2_v$ for all $v\in C$, or $\varphi^1_v<
\varphi^2_v$ for all $v\in C$, or $C = \{v\}$ and $\varphi_v^1 = \varphi^2_v$.}
\end{split}
\end{equation*}
This follows from \eqref{eq:cluster_swap_omega_probability} by
noting that,
\begin{equation*}
  \frac{h_{(v,w)}(\varphi^2_v, \varphi^1_w)\cdot h_{(v,w)}(\varphi^1_v,\varphi^2_w)}{h_{(v,w)}(\varphi^1_v, \varphi^1_w)\cdot
  h_{(v,w)}(\varphi^2_v,\varphi^2_w)} = \exp(U(\varphi_v^1 -
  \varphi_w^1) + U(\varphi_v^2 - \varphi_w^2) - U(\varphi_v^2 -
  \varphi_w^1) - U(\varphi_v^1 - \varphi_w^2)).
\end{equation*}
Writing
\begin{align*}
  &\varphi_v^2 - \varphi_w^1 = p(\varphi_v^1 - \varphi_w^1) +
  (1-p)(\varphi_v^2 - \varphi_w^2),\\
  &\varphi_v^1 - \varphi_w^2 = (1-p)(\varphi_v^1 - \varphi_w^1) +
  p(\varphi_v^2 - \varphi_w^2)
\end{align*}
where
\begin{equation*}
  p:=\frac{\varphi_w^2 - \varphi_w^1}{\varphi_v^1 - \varphi_v^2 +
\varphi_w^2 - \varphi_w^1},
\end{equation*}
so that if either $\varphi_v^1 \ge \varphi_v^2$ and $\varphi_w^1 \le
\varphi_w^2$, or $\varphi_v^1 \le \varphi_v^2$ and $\varphi_w^1 \ge
\varphi_w^2$ (but at least one inequality is strict) then
$p\in[0,1]$ whence convexity of $U$ implies that
\begin{equation*}
  U(\varphi_v^2 - \varphi_w^1) + U(\varphi_v^1 - \varphi_w^2) \le U(\varphi_v^1 -
  \varphi_w^1) + U(\varphi_v^2 - \varphi_w^2).
\end{equation*}
Sheffield \cite{sheffield2005random} used the cluster swapping idea in his investigation of the
translation-invariant gradient Gibbs measures of random surfaces, in the real- and
integer-valued cases,
with convex potentials (following Funaki-Spohn \cite{funaki1997motion} for the
real-valued case). Along the way, he uses cluster swapping in a
beautifully simple manner to obtain monotnicity in boundary
conditions and log-concavity of the single-site marginal
distributions in random surfaces with convex potentials~\cite[Section 8.2]{sheffield2005random}.

Lastly, we remark that cluster swapping may sometimes be given a fruitful interpretation as the swapping of disagreement sets between configurations of an extended model. Such an interpretation was provided in~\cite{aizenman2019exponential} for the random-field Ising model where, rather than defining $\omega$, configurations of the model were extended to continuous functions on the metric graph - the metric space obtained by replacing edges with continuous segments - and swaps were made on connected components of the disagreement set of two extended configurations.

{\bf Markov chains:} We point out that Markov chains in discrete time also fit the framework~\eqref{eq:measure_for_reflections}. For simplicity, we discuss the case where the state space of the chain is countable. Let $S$ be a finite or countable set (with the discrete sigma algebra) and let $P = (P(a,b))_{a,b\in S}$ be the transition probabilities for a Markov chain on $S$. Thus we assume that
\begin{equation*}
  \text{$P(a,b)\ge 0$ for all $a,b\in S$ and $\sum_{b\in S} P(a,b) = 1$ for all $a\in S$}.
\end{equation*}

A finite sequence $X_0,X_1, \ldots, X_n$ of $S$-valued random variables is a sample of the Markov chain if
\begin{equation}\label{eq:Markov_chain_sample}
  \P(X_0 = x_0, X_1 = x_1,\ldots, X_n = x_n) = \mu(x_0)P(x_0, x_1)\cdots P(x_{n-1}, x_n),
\end{equation}
for some initial probability distribution $\mu$ on $S$. The distribution of the sequence thus fits the framework~\eqref{eq:measure_for_reflections} with the graph $G = (V,E)$ having $V = \{0,1,\ldots, n\}$ and $E = \{\{0,1\}, \ldots, \{n-1,n\}\}$, with the single-site measures $\lambda_0 := \mu$ and $\lambda_j$ being the counting measure on $S$ for $1\le j\le n$ and the interaction functions $h_{(j-1,j)} = P$ for $1\le j\le n$.

Thus, if one has in hand a reflection $\tau:S\to S$, one may apply the general reflection transformation procedure to the Markov chain. In this setting, letting $V_0 := \{0\}$, the conditions \eqref{eq:tau_involution}, \eqref{eq:tau_preserves_lambda} and \eqref{eq:tau_preserves_h} defining a reflection translate to $\tau(\tau(a)) = a$ for $a\in S$ and $P(\tau(a),\tau(b)) = P(a,b)$ for $a,b\in S$. As already mentioned, the reflection principle for simple random walk (see \eqref{eq:reflection_principle}) can be seen as a special case of this setup.

\medskip
In most of the discussion above, the functions $h_{(v,w)}$ were chosen
the same for all $(v,w)\in \vec{E}$. We point out that non-homogeneous setups, with $h_{(v,w)}$ depending on $(v,w)$, may arise even in the investigation of homogeneous models, as in the discussion at the end of Section~\ref{sec:monotonicity_of_pair_correlations}. The use of non-homogeneous $(\lambda_v)$ may also be natural in some, otherwise homogeneous, contexts. For instance, the height function of the dimer model on the triangular lattice has the modulo $3$ of the heights fixed on each of the three sub-lattices, and this restriction may be imposed with a suitable choice of $(\lambda_v)$ (or, alternatively, with a suitable choice of $(h_{(v,w)})$).

The usefulness of the above discussion is demonstrated in the next sections where we prove our main theorems, with the proofs in Section~\ref{sec:sublevel_set_connectivity} and Section~\ref{sec:monotonicity_of_pair_correlations} being particularly short.

\section{Sublevel set connectivity for random surfaces}\label{sec:sublevel_set_connectivity}
In this section we prove Theorem~\ref{thm:reflection_principle}. As remarked there, the lower bound in \eqref{eq:barrier_ineq} is trivial so we focus here on
proving the upper bound and \eqref{eq:Lipschitz_barrier_ineq}.

As in the theorem, let $G=(V,E)$ be a finite connected graph, let $V_0\subseteq V$ be
non-empty, let $U$ be a potential satisfying the monotonicity
condition \eqref{eq:monotonicity_condition_surfaces} and the
assumption \eqref{eq:partition_function_condition} that
$\mu_{U,G,V_0}$ is well-defined. Let $\varphi$ be randomly sampled
from $\mu_{U,G,V_{0}}$. We first prove a more general inequality than the upper bound in \eqref{eq:barrier_ineq},
\begin{equation}\label{eq:barrier_ineq_more_general}
  \mathbb{P}(V_0\centernot{\xleftrightarrow{\varphi< m}} v, \varphi_v\in D)\le
\mathbb{P}(\varphi_v\in 2m-D)\quad \text{for all $v\in V,
  m\ge 0$ and Borel $D\subseteq\mathbb{R}$},
\end{equation}
where $2m-D:=\{2m-t\,\colon\,t\in D\}$. The upper bound in \eqref{eq:barrier_ineq} follows since
\begin{equation*}
  \mathbb{P}(V_0\centernot{\xleftrightarrow{\varphi< m}} v) = \mathbb{P}(V_0\centernot{\xleftrightarrow{\varphi< m}} v, \varphi_v<m) +
  \mathbb{P}(\varphi_v\ge m) \le 2\mathbb{P}(\varphi_v\ge m) = \mathbb{P}(|\varphi_v|\ge m).
\end{equation*}
We proceed to prove \eqref{eq:barrier_ineq_more_general}. Let $\tau_m$ be the reflection introduced in \eqref{eq:random_surfaces_reflection}. Let
$\omega:E\to\{0,1\}$ be sampled according to \eqref{eq:random_surface_omega_probability} and note that, as a consequence of
\eqref{eq:random_surface_omega_connection_restriction}, we have
\begin{equation*}
  \{V_0\centernot{\xleftrightarrow{\varphi< m}} v, \varphi_v\in D\}\subseteq\{v\centernot{\xleftrightarrow{\omega}}V_0\}\quad\pmod\P
\end{equation*}
where we write $A_1\subseteq A_2\pmod\P$, for events $A_1,A_2$, to indicate that $\P(A_1\setminus A_2)=0$. Thus
\begin{equation*}
  \{V_0\centernot{\xleftrightarrow{\varphi< m}} v, \varphi_v\in D\}\subseteq\{\varphi^{\omega,v}_v \in 2m-D\}\quad\pmod\P
\end{equation*}
where we recall the definition of $\varphi^{\omega, v}$ from \eqref{eq:flipped_configuration}. Thus,
\begin{equation*}
  \P(V_0\centernot{\xleftrightarrow{\varphi< m}} v, \varphi_v\in D)\le \P(\varphi^{\omega,v}_v \in 2m-D) = \P(\varphi_v\in 2m-D),
\end{equation*}
where the last equality follows from Lemma~\ref{lem:flip_preserves_ES_coupling}. This establishes \eqref{eq:barrier_ineq_more_general}.

Now suppose, additionally, that $U$ satisfies the finite-support condition
\eqref{eq:Lipschitz_condition_surfaces}. We again prove a more general inequality than needed,
\begin{equation}\label{eq:Lipschitz_barrier_ineq_more_general}
  \mathbb{P}(V_0\centernot{\xleftrightarrow{\varphi< m}} v, \varphi_v\in D)\ge
\mathbb{P}(\varphi_v\in 2m+1-D)\quad \text{for all $v\in V,
  m\ge 0$ and Borel $D\subseteq(-\infty, m]$}.
\end{equation}
The inequality \eqref{eq:Lipschitz_barrier_ineq} is a consequence of \eqref{eq:Lipschitz_barrier_ineq_more_general} and the symmetry of $\varphi_v$ as
\begin{equation*}
\begin{split}
  \mathbb{P}(V_0\centernot{\xleftrightarrow{\varphi< m}} v) &= \mathbb{P}(V_0\centernot{\xleftrightarrow{\varphi< m}} v, \varphi_v<m) +
  \mathbb{P}(\varphi_v\ge m) \ge \mathbb{P}(\varphi_v>m+1)+\mathbb{P}(\varphi_v\ge m)\\
  &= \mathbb{P}(|\varphi_v|\ge m) - \P(\varphi_v\in(m,m+1)).
\end{split}
\end{equation*}
To see \eqref{eq:Lipschitz_barrier_ineq_more_general}, we now consider the reflection $\tau_{m+\frac{1}{2}}$. Again, we let $\omega:E\to\{0,1\}$ be sampled
according to \eqref{eq:random_surface_omega_probability} (with respect to $\tau_{m+\frac{1}{2}}$) and note that, by
\eqref{eq:random_surface_omega_connection_restriction},
\begin{equation*}
  \{\varphi_v\ge m+1\}\subseteq\{v\centernot{\xleftrightarrow{\omega}}V_0\}\quad\pmod\P.
\end{equation*}
Additionally, the finite-support condition \eqref{eq:Lipschitz_condition_surfaces} implies that on the event $\{\varphi_v\ge m+1\}$ any path connecting $V_0$
and $v$ must pass through a vertex $w$ on which $\varphi_w\in [m,m+1]$. Since $\tau_{m+\frac{1}{2}}$ preserves the interval $[m, m+1]$ we conclude that
\begin{equation*}
  \{\varphi_v\ge m+1\}\subseteq\{V_0\centernot{\xleftrightarrow{\varphi^{\omega, v}< m}} v\}\quad\pmod\P.
\end{equation*}
Altogether, we conclude that, for each Borel $D\subseteq(-\infty,m]$,
\begin{equation*}
  \mathbb{P}(\varphi_v\in 2m+1-D) \le \mathbb{P}(V_0\centernot{\xleftrightarrow{\varphi^{\omega, v}< m}} v, \varphi^{\omega,v}_v\in D) =
  \mathbb{P}(V_0\centernot{\xleftrightarrow{\varphi< m}} v, \varphi_v\in D)
\end{equation*}
where we used Lemma~\ref{lem:flip_preserves_ES_coupling} in the last step. This concludes the proof of \eqref{eq:Lipschitz_barrier_ineq_more_general}.

\section{Monotonicity of pair correlations in spin $O(n)$ models}\label{sec:monotonicity_of_pair_correlations}
In this section we prove Theorem~\ref{thm:monotonicity_of_spin_O_n_density}. Let $v\in V\setminus V_0$. Let $b_1,b_2\in \mathbb{S}^{n-1}$ be such that $\langle b_1,e_1\rangle \geq\langle b_2, e_1\rangle$ and $b_1\neq  b_2$.   Define
\[
a := \frac{b_2-b_1}{||b_2-b_1||}
\]
 and consider the reflection along the hyperplane orthogonal to $a$, as defined in \eqref{eq:spin_system_reflection},
\[
\tau_a(b)=b-2\langle a,b\rangle\cdot a.
\]
As explained, $\tau_a$ is a reflection for the spin $O(n)$ model with potential $U$. In addition $\tau_a(b_2) = b_1$.
Moreover, we observe that
\begin{equation}\label{one_side}
\langle  e_1,a\rangle = ||b_2-b_1||^{-1}(\langle e_1,  b_2\rangle - \langle e_1,  b_1\rangle)\le 0,
\end{equation}
and that
\begin{equation*}
\langle b_2,a\rangle =||b_2-b_1||^{-1}(\langle b_2,b_2\rangle - \langle b_2,b_1\rangle)> 0.
\end{equation*}
Choose $r>0$ sufficiently small that
\begin{equation}\label{other_side}
\langle b, a\rangle > 0\quad\text{for all } b\in B(b_2,r),
\end{equation}
where we write $B(b_2, r):=\{b\in\mathbb{R}^n\,|\,\|b-b_2\|_2\le r\}$ for the closed Euclidean ball of radius $r$ around $b$.

As before, we let $\omega:E\to\{0,1\}$ be sampled
according to \eqref{eq:spin_On_omega_probability} (with respect to $\tau_a$). The relations \eqref{one_side} and \eqref{other_side} combined with the property
\eqref{spin_one_sided_reflection} imply that
\begin{equation*}
  \{\varphi_v\in B(b_2, r)\}\subseteq\{v\centernot{\xleftrightarrow{\omega}}V_0\}\quad\pmod\P.
\end{equation*}
Consequently, recalling the definition of $\varphi^{\omega, v}$ from \eqref{eq:flipped_configuration} and the fact that $\tau_a(b_2)=b_1$ and $\tau_a$ is an isometry,
\[
\{\varphi_v\in B(b_2, r) \} \subseteq \{\varphi^{\omega, v}_v\in B(b_1, r)\}\quad\pmod\P.
\]
Since $\varphi^{\omega,v}$ has the same distribution as $\varphi$, we have
\[
\mathbb{P}\left(\varphi_v\in B(b_2,r)\right)\le \mathbb{P}\left(\varphi^{\omega,v}_v\in B(b_1,r)\right) = \mathbb{P}\left(\varphi_v\in B(b_1,r)\right).
\]
Fix a version $d_v$ of the density of $\varphi_v$ with respect to $\mu_{\mathbb{S}^{n-1}}$, which we assume is rotationally symmetric in the sense that $d_v(b)$ is a function of $\langle b, e_1\rangle$. As
\begin{equation}\label{eq:density_limit}
  \frac{\mathbb{P}\left(\varphi_v\in B(b,r)\right)}{\mu_{\mathbb{S}^{n-1}}(B(b,r))}\to d_v(b)\quad\text{as $r\downarrow0$}
\end{equation}
for $\mu_{\mathbb{S}^{n-1}}$-almost every $b$, we conclude that $d_v(b)$ is a monotone increasing function of $\langle b, e_1\rangle$ except on the $\mu_{\mathbb{S}^{n-1}}$-null set of $b$'s for which the convergence in \eqref{eq:density_limit} fails. We may then redefine $d_v(b)$ on this null set to make $d_v(b)$ monotone for all $b$, as we wanted to show.

\medskip
We may prove Theorem~\ref{thm:non-increasing_density}, regarding the monotonicity of marginal densities in random surfaces, in exactly the same manner, replacing the reflection $\tau_a$ on $\mathbb{S}^{n-1}$ by the reflection $\tau_m$ on $\mathbb{R}$. For variety (and possible interest in other contexts), we briefly explain an alternative route to the proof of Theorem~\ref{thm:non-increasing_density} via convexity considerations (we are not aware of such an alternative for proving Theorem~\ref{thm:monotonicity_of_spin_O_n_density}).

When the potential $U$ is a convex function, the distribution of $\mu_{U,G,V_0}$ is log-concave and centrally symmetric (invariant to a global sign flip), whence the marginal distribution of $\varphi_x$ is also log-concave and centrally symmetric so that $|\varphi_x|$ has a non-increasing density. While our assumption that $U$ is monotone does not imply that $U$ is convex, it allows us to decompose the distribution of $\mu_{U,G,V_0}$ as a mixture of centrally-symmetric log-concave distributions and deduce Theorem~\ref{thm:non-increasing_density} as before. The decomposition is a special case of the Edwards-Sokal decomposition \cite{edwards1988generalization} and we briefly describe it next.

Let $t = (t_e)_{e\in E}\in (0,\infty)^E$ and define the measure $\mu_{t,G,V_0}$ by
\begin{equation}\label{hammock_measure}
 d\mu_{t,G,V_0}\left(\varphi\right):=\frac{1}{Z_{t,G,V_0}}\prod\limits _{\{v,w\}\in E}\mathbf{1}_{\left[-t_{\{v,w\}},t_{\{v,w\}}\right]}
 \left(\varphi_{v}-\varphi_{w}\right)\prod_{v\in V_0}\delta_{0}\left(d\varphi_{v}\right)\prod\limits _{v\in V\setminus
 V_0}d\varphi_{v},
\end{equation}
where $Z_{t,G,V_0}$ is a normalizing constant. In other words, the measure $\mu_{t,G,V_0}$ is uniform on the set of all $t$-Lipschitz functions, functions changing by at most $t_{\{v,w\}}$ on the edge $\{v,w\}$, normalized to be $0$ at $V_0$. In particular, $\mu_{t,G,V_0}$ is log-concave and centrally-symmetric. We say that $\mu_{t,G,V_0}$ is the measure of a \emph{random surface with inhomogeneous hammock potentials}.

Recalling our assumption that $U$ is monotone in the sense \eqref{eq:monotonicity_condition_surfaces}, let us suppose, for convenience, that $e^{-U}$ is right continuous on
$[0,\infty)$ (noting that the measure $\mu_{U,G,V_0}$ is invariant to
changing the value of $U$ at countably many points). Define the
measure $\lambda_U$ with the Lebesgue-Stieltjes differential,
$d\lambda_U(t):=-d\exp\left(-U(t)\right)$ on $[0,\infty)$, that is,
\[
\int\limits_{(a,b]}d\lambda
_U(s):=\exp\left(-U(a)\right)-\exp\left(-U(b)\right),~ 0<a<b,~and~
\lambda_{U}\left(\left\{0\right\}\right):=0.
\]
Since $U$ is symmetric, we can write
\[
\exp\left(-U\left(x\right)\right)=\int\limits
_{(|x|,\infty)}d\lambda_{U}(s)=
\int\limits_{[0,\infty)}\mathbf{1}_{\left[-s,s\right]}\left(x\right)d\lambda_{U}(s).
\]
Substituting this expression for $\exp\left(-U\left(x\right)\right)$ in the density of $\mu_{U,G,V_0}$ given in \eqref{eq:random_surface_measure} shows, after a short calculation, that $\mu_{U,G,V_0}$ is a mixture, with respect to $t$, of measures of the form $\mu_{t,G,V_0}$.

As remarked in Section~\ref{sec:spin_systems}, analogs of Theorem~\ref{thm:monotonicity_of_spin_O_n_density} and Theorem~\ref{thm:non-increasing_density} continue to hold for clock models and integer-valued random surfaces, with the same proofs. To prove the analog of Theorem~\ref{thm:non-increasing_density} via the convexity approach, one needs to know that if the potential $U$, now defined on integers, is convex in a suitable sense then the marginal distribution of $\varphi_x$ is log-concave (in a suitable sense). Such a result was proved by Sheffield \cite[Lemma
8.2.4]{sheffield2005random} (using the cluster-swapping method described in Section~\ref{sec:reflection_transformations_cluster_algorithms}).

\section{Estimating the probability of having extremal gradients}\label{extremal_gradients_sec}
In this section we prove Theorem~\ref{main_thm}.

Let $G=\left(V,E\right)$ be a finite connected graph and let
$V_0\subseteq V$ be non-empty.  Let $U$ be a potential satisfying assumptions
\eqref{eq:monotonicity_condition_surfaces} and \eqref{eq:Lipschitz_condition_surfaces}, and recall from
Section~\ref{sec:introduction} the random surface measure $\mu_{U,G,V_0}$. Observe that if any of the given edges $\{v_1,w_1\},\ldots,\{v_k, w_k\}$ has both endpoints in $V_0$ then the conclusion \eqref{eq:main_result} of Theorem~\ref{main_thm} follows trivially as the gradient of $\varphi$ on that edge is zero almost surely. We assume henceforth that each of the given edges has at most one endpoint in $V_0$. Without loss of generality, we now assume that $V_0$ is a singleton, that is,
\[
 V_0=\{v_0\}.
\]
as we can replace our graph with the graph in which all vertices in $V_0$ have been identified to a single vertex $v_0$, erasing self-loops and keeping a single representative of each multiple edge, and note that the random surface measure is naturally preserved under this operation. This identification operation may substantially increase the degree of $v_0$, possibly beyond the maximal degree in the original graph. Thus, we will take care not to rely on the degree of $v_0$ in our proofs. The degrees of all other vertices cannot increase under the identification operation. To address these issues it is convenient to define, for a given set of edges $F\subseteq E$,
\begin{equation*}
  \Delta(F) := \max\{\deg(w)\,:\, \exists \{v,w\}\in F, w\neq v_0\}
\end{equation*}
where we write $\deg(w)$ for the degree of $w$ in the graph $G$, so that $\Delta(F)$ is the maximal degree of a vertex other than $v_0$ in one of the edges of $F$. For brevity we write
\[
\mu_{U,G,v_0}:=\mu_{U,G,\{v_0\}}
\]

 We start with several definitions which will be used throughout the section. For a given
set of edges $H\subseteq E$, let
\begin{equation*}
  \orient(H) := \left\{(v,w):\left\{v,w\right\}\in H\right\}
\end{equation*}
stand for all orientations of the edges of $H$ (each undirected edge appears with both orientations in $\orient(H)$). For brevity, we
denote the set of all oriented edges by
\begin{equation*}
  \dvec{E}:=\orient(E).
\end{equation*}
We will frequently reference the event that the random surface has
extremal gradients on a given set of edges. This event will be used
both for oriented and for unoriented sets of edges and thus we
define, for each $0<\varepsilon<1$,
\begin{equation}\label{extreme_slopes}
\begin{split}
\LL(H,\varepsilon)&:= \left\{\left|\varphi_{v}-\varphi_{w}\right|\ge
1-\varepsilon \text{ for all } \left\{v,w\right\}\in H\right\},\quad
H\subseteq E,\\
\LL(\dvec{H},\varepsilon)&:=
\left\{\left|\varphi_{v}-\varphi_{w}\right|\ge 1-\varepsilon\text{
for all } (v,w)\in \dvec{H}\right\},\quad\dvec{H}\subseteq \dvec{E}.
\end{split}
\end{equation}
Theorem~\ref{main_thm} is proved as a consequence of three lemmas
which we now proceed to describe.

{\bf Diluting the given edge set}. We partition the real line into
$9$ sets, each of which is an arithmetic progression of intervals,
as follows
\begin{equation}\label{eq:D_j_def}
  D_{j}:=\frac{j}{4}+\left[-\frac{1}{8},\frac{1}{8}\right)+2\frac{1}{4}\mathbb{Z} = \left\{\frac{j}{4} + x + 2\frac{1}{4}k\;:\;-\frac{1}{8}\le x< \frac{1}{8},\,
  k\in\mathbb{Z}\right\},\quad
1\leq j\leq9.
\end{equation}
(the shorthand $k\frac{\ell}{m}$ means $k+\frac{\ell}{m}$).
The main property of these domains which we shall make use of is
that, for each $j$, the set $D_j$ is invariant to reflection with
respect to numbers in $\frac{j}{4} + 1\frac{1}{8} +
2\frac{1}{4}\mathbb{Z}$. In other words, for any $1\le j\le 9$,
\begin{equation}\label{eq:D_j_invariance}
  \text{if $y\in D_j$ then also $2m-y\in D_j$ for any $m\in \frac{j}{4} +
  1\frac{1}{8}+2\frac{1}{4}\mathbb{Z}$}.
\end{equation}
The sets $D_j$ are in fact invariant to reflections with respect to
numbers in the larger set $\frac{j}{4} + 1\frac{1}{8}\mathbb{Z}$ but
this will not be used in our proof.

Given a set of oriented edges, we define the event that the value
of the surface on the first vertex of each oriented edge belongs to
$D_j$,
\begin{equation*}
   \Omega_{j}(\dvec{H}):=\left\{\varphi_{v}\in D_{j} \text{ for all }
\left(v,w\right)\in\dvec{H}\right\},\quad 1\leq j\leq 9
\end{equation*}
and also
\begin{equation*}
   \Omega(\dvec{H}):=\bigcup_{j=1}^{9}\Omega_{j}(\dvec{H}).
\end{equation*}
We will make use of certain separation properties between oriented
edges as given in the following definition.
\begin{definition}\label{separated set} (Separated set). A subset of oriented
edges $\dvec{H}\subseteq\dvec{E}$ is said to be \emph{separated} if
\begin{enumerate}
\item[(i)] for every distinct $\left(v_{1},w_{1}\right),\left(v_{2},w_{2}\right)\in\dvec{H}$,
$w_1\neq w_2$, $v_1\neq w_2$ and $v_2\neq w_1$,
\item[(ii)] for all  $\left(v,w\right)\in\dvec{H},\;w\neq v_0$.
\end{enumerate}
\end{definition}
In words, a separated set of oriented edges is a set in which every
two edges are either disjoint or coincide in their first vertex and
in which no edge is oriented towards  $v_0$.

Our first lemma shows that the probability of $\LL(F,\varepsilon)$,
for a given $F\subseteq E$, may be bounded in terms of the
probability of $\LL(\dvec{H},\varepsilon)\cap\Omega(\dvec{H})$, for
some large separated $\dvec{H}\subseteq\orient(F)$.
\begin{lemma}\label{separated_set_lemma}
Let $0<\varepsilon<1$, let $F\subseteq E$ be a non-empty set of
edges and let $\varphi$ be randomly sampled from $\mu_{U,G,v_0}$. Then
there exists a separated set $\dvec{H}\subseteq\orient(F)$ satisfying
$|\dvec{H}|\geq \frac{|F|}{9\Delta(F)}$ and
\begin{equation}\label{properties}
\mathbb{P}\left(\LL(F,\varepsilon)\right)\leq
2^{|F|}\cdot\mathbb{P}\left(\LL(\dvec{H},\varepsilon)\cap\Omega(\dvec{H})\right).
\end{equation}
\end{lemma}

{\bf Unlocking edges}. We now describe the key step in our proof.
Suppose that the random surface $\varphi$ has an extremal gradient
on the oriented edge $\vec{e} = (v,w)$, oriented towards $w$, say,
in the sense that $\varphi_w - \varphi_v \ge 1-\varepsilon$. If
$\varphi_u$ is not much higher than $\varphi_w$ for all neighbors
$u$ of $w$, then we may change the value of the surface on $w$,
thereby reducing the gradient on $(v,w)$ without reducing
significantly the density under the random surface measure (see
Lemma~\ref{lem:prob_of_extremal_unlocked_edge} below). Therefore,
the difficulty in showing that extremal gradients are rare lies in
the possibility that the edge with the extremal gradient is being
`locked' into this extreme position by a neighbor of one of its
endpoints. Our next definition quantifies the notion that an edge is
not locked in such a manner.
\begin{definition}\label{def:unlock_event}
An oriented edge  $\vec{e}=(v,w)\in\dvec{E}$ is called
\emph{unlocked} in $\varphi\in\mathbb{R}^V$ if
\begin{equation}
\max\limits_{u:\{u,w\}\in E}|\varphi_u-\varphi_v|\leq 1\frac{1}{4}.
\end{equation}
We define the corresponding event as
\begin{equation*}
\UL_{\vec{e}}:=\left\{\vec{e} \text{  is  unlocked in } \varphi
\right\}.
\end{equation*}
\end{definition}
The following key lemma will allow us to reduce our study of
extremal edges to the case that these edges are unlocked.
\begin{lemma}\label{UL_operator1} Let $0<\varepsilon<1$, let $\dvec{H}\subseteq \dvec{E}$ be a non-empty, separated set of oriented
edges, let $\vec{e} = (v,w)\in\dvec{H}$ and let $\varphi$ be randomly
sampled from $\mu_{U,G,v_0}$. Then
\[
\mathbb{P}\left(\LL(\dvec{H},\varepsilon)\cap\Omega(\dvec{H})\right)\leq2^{\deg(w)-1}\cdot\mathbb{P}\left(\LL(\dvec{H},\varepsilon)\cap\Omega(\dvec{H})\cap\UL_{\vec{e}}\right).
\]
\end{lemma}
The proof of the lemma uses the reflection
transformation described in
Section~\ref{sec:reflection_transformations_cluster_algorithms}.

{\bf The probability that an unlocked edge is extremal}. As noted
above, if an edge $(v,w)\in\dvec{E}$ has an extremal gradient and is
unlocked in the surface $\varphi$, then we may change $\varphi_w$ to
reduce the gradient on the edge while controlling the change in the
density of the surface under the measure $\mu_{U,G,v_0}$. This idea is
quantified by the next lemma.
\begin{lemma}\label{lem:prob_of_extremal_unlocked_edge}
  Let $0<\varepsilon\le \frac{1}{8}$, let $\vec{e} =
  (v,w)\in\dvec{E}$ with $w\neq v_0$ and let $\varphi$ be randomly
sampled from $\mu_{U,G,v_0}$. Then
  \begin{equation*}
    \mathbb{P}\left(|\varphi_w - \varphi_v|\ge 1-\varepsilon\,\,|\,\,
  (\varphi_u)_{u\in V\setminus\{w\}}\right)\cdot \mathbf{1}_{\UL_{\vec{e}}} \le
  \delta(U,\vec{e},\varepsilon),\quad\text{almost surely},
  \end{equation*}
  where $\mathbf{1}_A$ denotes the indicator random variable of the event $A$ and where we write
  \begin{equation*}
    \delta(U,\vec{e},\varepsilon) := 8\varepsilon\cdot \exp\left(-U(1-\varepsilon)+U(0)+\deg(w)\left(U\left(\frac{3}{4}\right) - U(0)\right)\right).
  \end{equation*}
\end{lemma}

The above three lemmas are proved in the next section. We now
explain how Theorem~\ref{main_thm} follows as a consequence of them.

{\bf Proof of Theorem~\ref{main_thm}}. Let
$0<\varepsilon\le\frac{1}{8}$, let $F\subseteq E$ be a non-empty set
of edges and let $\varphi$ be randomly sampled from $\mu_{U,G,v_0}$. Our goal is to estimate the probability of $\LL(F,\varepsilon)$.

Using Lemma~\ref{separated_set_lemma}, let
$\dvec{H}\subseteq\orient(F)$ be a separated set satisfying
$|\dvec{H}|\geq \frac{|F|}{9\Delta(F)}$ and
\begin{equation}\label{eq:H_from_F}
  \mathbb{P}\left(\LL(F,\varepsilon)\right)\leq
2^{|F|}\mathbb{P}\left(\LL(\dvec{H},\varepsilon)\cap\Omega(\dvec{H})\right).
\end{equation}

Let $\vec{e}=(v,w)\in\dvec{H}$. By Lemma~\ref{UL_operator1}, we have
\begin{equation}\label{eq:extremal_and_separated_estimate1}
\begin{split}
\mathbb{P}(\LL(\dvec{H},\varepsilon)\cap\Omega(\dvec{H}))&\leq2^{\deg(w)-1}\mathbb{P}\left(\LL(\dvec{H},\varepsilon)\cap\Omega(\dvec{H})\cap\UL_{\vec{e}}\right)\\
&\le 2^{\deg(w)-1}\mathbb{P}\left(\{|\varphi_v-\varphi_w|\ge
1-\varepsilon\}\cap\LL(\dvec{H}\setminus\{\vec{e}\},\varepsilon)\cap\Omega(\dvec{H}\setminus\{\vec{e}\})\cap\UL_{\vec{e}}\right).
\end{split}
\end{equation}
We note that, as $\dvec{H}$ is a separated set, the events
$\LL(\dvec{H}\setminus\{\vec{e}\},\varepsilon)$,
$\Omega(\dvec{H}\setminus\{\vec{e}\})$ and $\UL_{\vec{e}}$ are
measurable with respect to the random variables $(\varphi_u)_{u\in
V\setminus\{w\}}$. Thus, using
Lemma~\ref{lem:prob_of_extremal_unlocked_edge}, we may estimate
\begin{equation}\label{eq:extremal_and_separated_estimate2}
\begin{split}
  &\mathbb{P}\left(\{|\varphi_v-\varphi_w|\ge
1-\varepsilon\}\cap\LL(\dvec{H}\setminus\{\vec{e}\},\varepsilon)\cap\Omega(\dvec{H}\setminus\{\vec{e}\})\cap\UL_{\vec{e}}\right)\\
  &=\mathbb{E}\left[\mathbb{P}\left(|\varphi_w - \varphi_v|\ge 1-\varepsilon\,\,|\,\,
  (\varphi_u)_{u\in
  V\setminus\{w\}}\right)\mathbf{1}_{\LL(\dvec{H}\setminus\{\vec{e}\},\varepsilon)\cap\Omega(\dvec{H}\setminus\{\vec{e}\})\cap\UL_{\vec{e}}}\right]\\
  &\le
  \delta(U,\vec{e},\varepsilon)\mathbb{P}\left(\LL(\dvec{H}\setminus\{\vec{e}\},\varepsilon)\cap\Omega(\dvec{H}\setminus\{\vec{e}\})\right).
\end{split}
\end{equation}
Recall from the statement of Theorem~\ref{main_thm} that
\begin{equation*}
  \delta(U, \varepsilon)=\varepsilon\cdot \exp\left(-U(1-\varepsilon)+U(0)+\Delta(F)\left(U\left(\frac{3}{4}\right) - U(0)\right)\right)
\end{equation*}
and observe that, as $w\neq v_0$ since $\dvec{H}$ is separated and as
$U(3/4)\ge U(0)$ by our assumption that $U$ is non-decreasing on
$[0,\infty)$,
\begin{equation}\label{eq:delta_U_epsilon_comparison}
  \deg(w)\le \Delta(F)\quad\text{and}\quad\delta(U,\vec{e},\varepsilon)\le 8\delta(U, \varepsilon).
\end{equation}
Putting together \eqref{eq:extremal_and_separated_estimate1},
\eqref{eq:extremal_and_separated_estimate2} and
\eqref{eq:delta_U_epsilon_comparison}, we conclude that
\begin{equation*}
  \mathbb{P}(\LL(\dvec{H},\varepsilon)\cap\Omega(\dvec{H}))\le 2^{\Delta(F)-1}\cdot8\delta(U,
  \varepsilon)\mathbb{P}\left(\LL(\dvec{H}\setminus\{\vec{e}\},\varepsilon)\cap\Omega(\dvec{H}\setminus\{\vec{e}\})\right).
\end{equation*}
Iterating this estimate over all edges in $\dvec{H}$ shows that
\begin{equation*}
  \mathbb{P}(\LL(\dvec{H},\varepsilon)\cap\Omega(\dvec{H}))\le
  \left(2^{\Delta(F)+2}\delta(U,\varepsilon)\right)^{|\dvec{H}|}.
\end{equation*}
Substituting this estimate into \eqref{eq:H_from_F} and using the
fact that $|\dvec{H}|\geq \frac{|F|}{9\Delta(F)}$, we have
\begin{equation*}
  \mathbb{P}\left(\LL(F,\varepsilon)\right)\leq
\min\left(2^{|F|}\left(2^{\Delta(F)+2}\delta(U,\varepsilon)\right)^{|\dvec{H}|},1\right)\le
\min\left(\left(2^{10\Delta(F)+2}\delta(U,\varepsilon)\right)^{\frac{|F|}{9\Delta(F)}},
1\right).
\end{equation*}
This concludes the proof of Theorem~\ref{main_thm}, given the above
lemmas, with the constant $C(\Delta) = 2^{10\Delta +2}$.

\subsection{Proof of Lemma~\ref{separated_set_lemma}}

\begin{proof}
Assign an orientation to the edges of $F$, chosen arbitrarily except for the rule that edges having $v_0$ as an endpoint are oriented to have $v_0$ as their first vertex. Denote the resulting set of oriented edges by $\dvec{F}$.

Let $\varphi$ be randomly sampled from $\mu_{U,G,v_0}$. Observe that there
is a $1\le j\le 9$ and a (random) subset $\dvec{F}'\subseteq\dvec{F}$
satisfying $|\dvec{F}'|\ge \frac{1}{9}|\dvec{F}|$ such that
$\varphi_{v}\in D_{j}$ for all $(v,w)\in\dvec{F}'$ (as each edge
satisfies $\varphi_{v}\in D_{j}$ for a unique $j$). In addition, we
may choose a (random) separated subset $\dvec{H}\subseteq\dvec{F}'$
satisfying $|\dvec{H}|\ge \frac{1}{\Delta(F)}|\dvec{F}'|$. Indeed, this
may be done in a greedy manner: Sequentially, for each edge $(v,w)$
still in $\dvec{F}'$, we discard from $\dvec{F}'$ all edges $(x,y)$
with either $x=w$ or $y=w$, discarding in this way at most
$\deg(w)-1\le \Delta(F)-1$ edges. In
conclusion, defining
\begin{equation}
\mathcal{H}:=\left\{\dvec{H}\subseteq\dvec{F}:|\dvec{H}|\geq
\frac{|F|}{9\Delta(F)},\ \dvec{H}\text{ is separated}\right\}
\end{equation}
we have shown that, almost surely, $\varphi\in\Omega(\dvec{H})$ for some $\dvec{H}\in\mathcal{H}$. Thus, using that $|\mathcal{H}|\le 2^{|\dvec{F}|}=2^{|F|}$,
\[
\mathbb{P}(\LL(F,\varepsilon))\leq\sum\limits_{\dvec{H}\in\mathcal{H}}\mathbb{P}(\LL(\dvec{H},\varepsilon)\cap \Omega(\dvec{H}))\leq
2^{|F|}\max\limits_{\dvec{H}\in\mathcal{H}}\mathbb{P}\left(\LL(\dvec{H},\varepsilon)\cap\Omega(\dvec{H})\right).\qedhere
\]
\end{proof}

\subsection{Proof of Lemma~\ref{UL_operator1}}
Fix $0<\varepsilon<1$, a non-empty, separated set of oriented edges
$\dvec{H}\subseteq \dvec{E}$ and an oriented edge $\vec{e} =
(v,w)\in\dvec{H}$.
Let $\varphi$ be randomly sampled from $\mu_{U,G,v_0}$. As the events $\Omega_j(\dvec{H})$ which comprise $\Omega(\dvec{H})$ are disjoint, it suffices to prove that
\begin{equation}\label{eq:UL_op_lemma_inequality_with_j}
  \mathbb{P}\left(\LL(\dvec{H},\varepsilon)\cap\Omega_j(\dvec{H})\right)\leq2^{\deg(w)-1}\cdot\mathbb{P}\left(\LL(\dvec{H},\varepsilon)\cap\Omega_j(\dvec{H})\cap\UL_{\vec{e}}\right),\quad 1\le j\le 9.
\end{equation}
Thus we also fix $1\le j\le 9$. Define the subset $M_j$ of the real line by
\begin{equation*}
  M_j := \frac{j}{4} + 1\frac{1}{8}+2\frac{1}{4}\mathbb{Z}.
\end{equation*}
Now define
\begin{equation*}
  m(\varphi) := \begin{cases}
    \min(m\in M_j : m\ge\varphi_v) & \varphi_w \ge \varphi_v\\
    \max(m\in M_j : m\le\varphi_v) & \varphi_w < \varphi_v
  \end{cases},
\end{equation*}
so that $m(\varphi)$ is on the same side of $\varphi_v$ as $\varphi_w$ and $m(\varphi)\neq \varphi_v$ on the event $\Omega_j(\dvec{H})$. Define also
\begin{equation}\label{def:locked_set}
W(\varphi):=\left\{u\in V\, :\,\{u,w\}\in E,\,\sign(\varphi_{u}-m(\varphi))=\sign(m(\varphi) - \varphi_v)\right\},
\end{equation}
where, as usual, $\sign(x)=1$ if $x>0$, $\sign(x)=-1$ if $x<0$ and $\sign(x)=0$ if $x=0$.
We shall prove that for each $W\subseteq V$ and $m\in M_j$,
\begin{equation}\label{eq:W_varphi_inequality}
  \mathbb{P}(\LL(\dvec{H},\varepsilon)\cap\Omega_j(\dvec{H})\cap\{W(\varphi) = W\}\cap\{m(\varphi) = m\}) \le \mathbb{P}(\LL(\dvec{H},\varepsilon)\cap\Omega_j(\dvec{H})\cap
\UL_{\vec{e}}\cap\{m(\varphi) = m\}).
\end{equation}
This relation implies \eqref{eq:UL_op_lemma_inequality_with_j}, and hence the lemma, by summing over all
possible values of $W$ and $m$, and using the fact that
$\mathbb{P}(W(\varphi) = W) = 0$ unless $W$ is a subset of $\{u\in V\,:\,
\{u,w\}\in E,\, u\neq v\}$.

We proceed to prove \eqref{eq:W_varphi_inequality} and it is here that we make use of the reflection transformation described in Section~\ref{sec:reflection_transformations_cluster_algorithms}. Fix $W\subseteq V$ and $m\in M_j$ satisfying
\begin{equation*}
  \mathbb{P}(W(\varphi) = W, m(\varphi) = m)>0
\end{equation*}
(as the relation \eqref{eq:W_varphi_inequality} is trivial if this probability is zero). Recall the `reflection around $m$' mapping $\tau_m:\mathbb{R}\to\mathbb{R}$ given by $\tau_m(a) = 2m-a$, as in \eqref{eq:random_surfaces_reflection}. Let $(\varphi, \omega)$ be randomly sampled from the $\tau_{m}$-Edwards-Sokal coupling defined in \eqref{eq:random_surface_omega_probability}. We define the reflected configuration $\varphi^{\omega,W}:V\to\mathbb{R}$ as follows: If $W = \emptyset$ then we set $\varphi^{\omega,W}:=\varphi$. Otherwise,
\begin{equation}
  \text{if $W \centernot{\xleftrightarrow{\omega}} v_0$ then $\varphi^{\omega,W}_v := \begin{cases}
    \tau_m(\varphi_v)& W
\xleftrightarrow{\omega} v\\
    \varphi_v& W
\centernot{\xleftrightarrow{\omega}} v
  \end{cases}$.\quad  If $W \xleftrightarrow{\omega} v_0$ then $\varphi^{\omega,W}_v:=\begin{cases}
 \varphi_v& W
\xleftrightarrow{\omega} v\\
   \tau_m(\varphi_v)& W
\centernot{\xleftrightarrow{\omega}} v
  \end{cases}$}.
\end{equation}
It then follows from the discussion after Lemma \ref{lem:flip_preserves_ES_coupling} that $(\varphi^{\omega,W},\omega)$ has the same distribution as $(\varphi,\omega)$. The equality in distribution shows that \eqref{eq:W_varphi_inequality} is a consequence of the following relation:
\begin{multline}\label{eq:reflection_containment}
  \{\varphi\in \LL(\dvec{H},\varepsilon)\cap\Omega_j(\dvec{H})\}\cap\{W(\varphi) =
  W\}\cap\{m(\varphi) = m\}\subseteq\\ \{\varphi^{\omega,W} \in\LL(\dvec{H},\varepsilon)\cap\Omega_j(\dvec{H})\cap\UL_{\vec{e}}\}\cap\{m(\varphi^{\omega,W}) = m\}\quad\pmod\P
\end{multline}
(as before, we write $\!\!\!\pmod\P$ to indicate that the containment is in the sense of the difference having zero probability), where, with a slight abuse of notation, we consider the events $\LL(\dvec{H},\varepsilon)$, $\Omega_j(\dvec{H})$ and $\UL_{\vec{e}}$ as subsets of the space $\mathbb{R}^V$ of configurations. Thus, it remains to prove \eqref{eq:reflection_containment}, which is a consequence of the following three claims.
\begin{claim}\label{claim_1}  Almost surely, if $\varphi\in \LL(\dvec{H},\varepsilon)\cap\Omega_j(\dvec{H})$ then $\varphi^{\omega,W}\in \LL(\dvec{H},\varepsilon)$.
\end{claim}

\begin{proof}
We will prove the stronger consequence that, under the given assumptions,
\begin{equation}\label{eq:slope_is_preserved}
  |\varphi_x - \varphi_y| = |\varphi^{\omega,W}_x -
  \varphi^{\omega,W}_y|\quad\text{for all
  $(x,y)\in\dvec{H}$}.
\end{equation}
Fix $(x,y)\in\dvec{H}$. Observe that, by
definition of the reflection operation,
\begin{equation*}
  |\varphi^{\omega,W}_x -
  \varphi^{\omega,W}_y|\in \{|\varphi_x - \varphi_y|, |2m
  - \varphi_x - \varphi_y|\}.
\end{equation*}
Suppose that
\begin{equation}\label{eq:second_slope_case}
  |\varphi^{\omega,W}_x -
  \varphi^{\omega,W}_y|=
|2m - \varphi_x - \varphi_y|
\end{equation}
as in the other case \eqref{eq:slope_is_preserved} is clearly
satisfied. Since $\varphi\in\Omega_j(\dvec{H})$ it follows that
$\varphi_x\in D_{j}$. Thus, recalling the definition
\eqref{eq:D_j_def} of $D_j$ and the fact that $m\in M_j$ we see that
\begin{equation}\label{eq:d_m_varphi_x}
  |m - \varphi_x| \ge \dist(M_j, D_j) = 1.
\end{equation}
Since $|\varphi_x-\varphi_y|\le1$ it implies that
either $m\ge \max(\varphi_x, \varphi_y)$ or $m\le \min(\varphi_x,
\varphi_y)$. In both cases,
\begin{equation}\label{eq:triangle_equality}
  |2m - \varphi_x - \varphi_y| = |m - \varphi_x| + |m - \varphi_y| \ge 1.
\end{equation}
Since $\varphi^{\omega,W}$ is a Lipschitz function almost surely, we conclude from
\eqref{eq:second_slope_case} that equality must hold in
\eqref{eq:triangle_equality}. Taking into account
\eqref{eq:d_m_varphi_x}, this implies that $\varphi_y = m$ and $|m -
\varphi_x| = 1$, in which case $|2m - \varphi_x - \varphi_y| =
|\varphi_x - \varphi_y|$ so that \eqref{eq:slope_is_preserved}
holds.
\end{proof}
\begin{claim}\label{claim_2}
  If $\varphi\in \Omega_j(\dvec{H})$ then $\varphi^{\omega,W}\in \Omega_j(\dvec{H})$.
\end{claim}
\begin{proof} The claim follows from the fact that $\tau_m(D_j) \subseteq D_j$ since $m\in M_j$, as noted in \eqref{eq:D_j_invariance}.
\end{proof}
 \begin{claim}\label{claim_3}
  Almost surely, if $\varphi\in \LL(\dvec{H},\varepsilon)\cap\Omega_j(\dvec{H})$, $W(\varphi) =
  W$ and $m(\varphi) = m$ then $m(\varphi^{\omega,W}) = m$ and $\varphi^{\omega,W}\in \UL_{\vec{e}}$.
 \end{claim}
 \begin{proof}
Let $k\in\mathbb{Z}$ be such that $\varphi_v\in \frac{j}{4}+\left[-\frac{1}{8},\frac{1}{8}\right) + 2\frac{1}{4}k$, using that $\varphi\in \Omega_j(\dvec{H})$. As $\varphi\in \LL(\dvec{H},\varepsilon)$ we have that $\varphi_w\neq \varphi_v$. For concreteness, assume that $\varphi_w>\varphi_v$ with the other case being treated similarly. Thus $m = \frac{j}{4}+ 1\frac{1}{8} + 2\frac{1}{4}k$ and note that $\varphi_w\le m$ as $\varphi$ is a Lipschitz function.

The definition of $\varphi^{\omega,W}$ implies that $\varphi^{\omega,W}_v\in \{\varphi_v, 2m-\varphi_v\}$ and $\varphi^{\omega,W}_w\in \{\varphi_w, 2m-\varphi_w\}$. The fact that both $\varphi_v\le m$ and $\varphi_w\le m$ imply that in all four possibilities for the values of $\varphi^{\omega,W}_v$ and $\varphi^{\omega,W}_w$ we have $m(\varphi^{\omega,W}) = m$.

Fix $u$ with $\{u,w\}\in E$. If $u\in W$, that is $\varphi_u>m$, the definition of $\varphi^{\omega,W}$ and property \eqref{eq:random_surface_omega_connection_restriction} imply that, almost surely,
\begin{equation*}
\begin{split}
&\text{if $W \centernot{\xleftrightarrow{\omega}} v_0$ then $\varphi^{\omega,W}_v=\varphi_v,\; \varphi^{\omega,W}_u=2m-\varphi_u$,\quad and}\\
&\text{if $W \xleftrightarrow{\omega} v_0$ then $\varphi^{\omega,W}_v=2m-\varphi_v,\; \varphi^{\omega,W}_u=\varphi_u$}.
\end{split}
\end{equation*}
In both cases
  \[
  |\varphi_{v}^{\omega,W}- \varphi_{u}^{\omega,W}|=|\varphi_v +\varphi_u -2m|\le\max(m - \varphi_v, \varphi_u - m)\le 1\frac{1}{4}
  \]
where we used that $\varphi_u-m \le \varphi_u - \varphi_w\le 1$ as $\varphi$ is a Lipschitz function.

If $u\notin W$, that is $\varphi_u\le m$, then $\varphi_u-\varphi_v\in (-1, 1\frac{1}{4}]$ as $\varphi$ is a Lipschitz function. As $\varphi_v, \varphi_u\le m$ we conclude from the definition of $\varphi^{\omega,W}$, the fact that $\varphi_x>m$ for all $x\in W$ and property \eqref{eq:random_surface_omega_connection_restriction} that, almost surely, $|\varphi_{u}^{\omega,W}- \varphi_{v}^{\omega,W}| = |\varphi_{u}- \varphi_{v}|$, whence $|\varphi_{u}^{\omega,W}- \varphi_{v}^{\omega,W}|\le 1\frac{1}{4}$.

Thus, $\max\limits_{u:\{u,w\}\in E} |\varphi^{\omega,W}_u-\varphi^{\omega,W}_v|\le1\frac{1}{4}$, from which it follows that $\varphi^{\omega,W}\in \UL_{\vec{e}}$.\qedhere
\end{proof}

\subsection{Proof of Lemma~\ref{lem:prob_of_extremal_unlocked_edge}}

Fix $0<\varepsilon\le\frac{1}{8}$ and $\vec{e}=(v,w)\in\dvec{E}$ with $w\neq v_0$. Let $\varphi$ be randomly sampled from $\mu_{U,G,v_0}$. The conditional density of $\varphi_w$ given $(\varphi_u)_{u\in V\setminus\{w\}}$ equals
\begin{equation*}
  \frac{\exp(-\sum_{u: \{u,w\}\in E} U(\varphi_u - \varphi_w))}{\int_{-\infty}^{\infty} \exp(-\sum_{u: \{u,w\}\in E} U(\varphi_u - x))dx}
\end{equation*}
Thus, the lemma will follow by showing that
\begin{equation*}
  \int_{(-\infty, -(1-\varepsilon)]\cup[1-\varepsilon, \infty)} e^{-\sum_{u: \{u,w\}\in E} U(\varphi_u - (\varphi_v +t))}dt \cdot\mathbf{1}_{U_{\vec{e}}}
  \le \delta(U,\vec{e},\varepsilon)\int_{-\infty}^{\infty} e^{-\sum_{u: \{u,w\}\in E} U(\varphi_u - (\varphi_v + t))}dt.
\end{equation*}
Taking into account the Lipschitz assumption~\eqref{eq:Lipschitz_condition_surfaces}, we see it suffices to prove the pair of inequalities,
\begin{align}
  &\int_{-1}^{-(1-\varepsilon)}e^{-\sum_{u: \{u,w\}\in E} U(\varphi_u - (\varphi_v +t))}dt \cdot\mathbf{1}_{U_{\vec{e}}}
  \le \delta(U,\vec{e},\varepsilon)\int_{-(3/4-\varepsilon)}^{-1/2}e^{-\sum_{u: \{u,w\}\in E} U(\varphi_u - (\varphi_v + t))}dt,\label{eq:cond_w_first}\\
  &\int_{1-\varepsilon}^{1}e^{-\sum_{u: \{u,w\}\in E} U(\varphi_u - (\varphi_v +t))}dt \cdot\mathbf{1}_{U_{\vec{e}}}
  \le \delta(U,\vec{e},\varepsilon)\int_{1/2}^{3/4-\varepsilon}e^{-\sum_{u: \{u,w\}\in E} U(\varphi_u - (\varphi_v + t))}dt.\label{eq:cond_w_second}
\end{align}
We prove only inequality~\eqref{eq:cond_w_second} as inequality~\eqref{eq:cond_w_first} follows from it by applying a global sign change to $\varphi$. We assume that $U_{\vec{e}}$ holds as \eqref{eq:cond_w_second} is trivially verified otherwise. In addition, we assume that
\begin{equation}\label{eq:u_cond_below}
\min\{\varphi_u\,:\, \{u,w\}\in E\}\ge \varphi_v - \varepsilon
\end{equation}
as otherwise, using the Lipschitz assumption~\eqref{eq:Lipschitz_condition_surfaces}, the left-hand side of \eqref{eq:cond_w_second} is zero, again verifying \eqref{eq:cond_w_second} trivially. We proceed to estimate separately the two integrals in \eqref{eq:cond_w_second}. First,
using the assumption~\eqref{eq:monotonicity_condition_surfaces} that $U$ is non-decreasing on $[0,\infty)$,
\begin{equation}\label{eq:left_hand_side_second_integral}
  \int_{1-\varepsilon}^{1} \exp\bigg(-\sum_{u:\{u,w\}\in E} U(\varphi_u - (\varphi_v+t))\bigg)dt \le
  \varepsilon \exp(-U(1-\varepsilon)-(\deg(w)-1)U(0)).
\end{equation}
Second, observe that $\max\limits_{u:\{u,w\}\in E}|\varphi_u-\varphi_v|\leq 1\frac{1}{4}
$ as $\UL_{\vec{e}}$ holds and therefore, using also~\eqref{eq:u_cond_below},
\[
\varphi_w -\varphi_v\in \left[1/2, 3/4-\varepsilon\right] \text{ implies that }\max\limits_{u:\{u,w\}\in E}|\varphi_u-\varphi_w|\le \frac{3}{4} .
\]
Using again the non-decreasing assumption~\eqref{eq:monotonicity_condition_surfaces} and the assumption that $\varepsilon \le \frac{1}{8}$, we obtain
\begin{equation}\label{eq:right_hand_side_second_integral}
\int_{1/2}^{3/4-\varepsilon} \exp\bigg(-\sum_{u:\{u,w\}\in E} U(\varphi_u -
  (\varphi_v+t))\bigg)dt\ge
  \frac{1}{8}\exp\left(-\deg(w)U\left(\frac{3}{4}\right)\right).
\end{equation}
Plugging the inequalities~\eqref{eq:left_hand_side_second_integral} and \eqref{eq:right_hand_side_second_integral} into \eqref{eq:cond_w_second} and comparing with the definition of $\delta(U,\vec{e},\varepsilon)$ verifies the inequality~\eqref{eq:cond_w_second} and finishes the proof of the lemma.

\section{Discussion and open questions}\label{sec:discussion_and_open_questions}
{\bf Extremal gradients.} Theorem~\ref{main_thm} provides quantitative estimates on the rarity of extremal gradients in random surfaces satisfying a Lipschitz constraint and having a monotone interaction potential. Its proof makes use of a cluster algorithm for random surfaces and thus provides an alternative to a previous approach via reflection positivity~\cite[Theorem~3.2]{milos2015delocalization}. The main advantage of the cluster algorithm approach is that it applies to random surfaces defined on general graphs and thus removes a chief limitation of the reflection positivity proof which is restricted to torus graphs. However, the proof presented in this paper introduces new limitations which it would be desirable to remove. Specifically, the current proof applies only to Lipschitz surfaces with monotone interaction potential whereas one may expect, as put forth explicitly in \cite[Section 6]{milos2015delocalization}, that results analogous to Theorem~\ref{main_thm} should hold almost without restriction on the potential function (some integrability conditions are required for the model to be well defined) and such a result indeed holds on torus graphs as shown with the reflection positivity method \cite[Theorem~3.2]{milos2015delocalization}.

Theorem~\ref{main_thm} may be used together with the arguments of \cite{milos2015delocalization} to prove the delocalization of random surfaces in cases not previously known. For instance, delocalization would follow for random surfaces whose potential satisfies \eqref{eq:monotonicity_condition_surfaces} and
\eqref{eq:Lipschitz_condition_surfaces} (such as the hammock potential~\eqref{eq:hammock_potential}) on finite connected domains of $\mathbb{Z}^2$ with Dirichlet boundary conditions (when $V_0$ is the set of boundary vertices of the domain), or any other choice of the non-empty normalization set $V_0$. The same techniques should apply to show delocalization on many (finite domains in) infinite graphs on which simple random walk is recurrent.

{\bf Extremal gradients in spin systems.} Similarly to the previous point, it is also of interest to extend the control of extremal gradients to the spin system setting. Bricmont and Fontaine \cite{bricmont1981correlation} show that extremal gradients are unlikely in the XY (spin $O(2)$) model (allowing even for multi-body interactions and external magnetic fields). Their proof makes use of Ginibre's extension of Griffiths' inequalities \cite{Gri67, Gin70} and thus does not extend to the spin $O(n)$ model with $n\ge 3$, where they obtain somewhat weaker results instead. As cluster algorithms are available in some generality for spin $O(n)$ models (as reviewed in Section~\ref{sec:reflection_transformations_cluster_algorithms}), it is possible that our approach to the control of extremal gradients may be extended to the spin system setting and provide additional results for models in which Ginibre's inequality is unavailable.

{\bf Excursion-set percolation.}
The reflection principle for random surfaces given in Theorem~\ref{thm:reflection_principle} may remind the reader of the study of excursion-set
percolation in random surfaces. Initiated by Lebowitz-Saleur \cite{lebowitz1986percolation}
and Bricmont-Lebowitz-Maes \cite{bricmont1987percolation}, this line of investigation focuses
on the percolative properties of the set $\{v\in V\,:\, \varphi_v\ge
h\}$. Triggered by its relations with the random interlacement model
\cite{sznitman2010vacant, rodriguez2013phase} introduced by Sznitman, the subject has recently seen significant activity; see \cite{Drewitz2017}
and references within. In these studies, one starts with the infinite-volume limit of a random surface $\varphi$ on $\mathbb{Z}^d$, typically the Gaussian free field with Dirichlet boundary conditions (see \cite{bricmont1987percolation, rodriguez2016decoupling} for an exception), and aims to study the set of $h\in\mathbb{R}$ for which there is, almost surely, percolation in the set of vertices $v$ with $\varphi_v\ge h$ (that is, there is an infinite connected component of vertices $v$ with $\varphi_v\ge h$). For the Gaussian free field, it is known \cite{bricmont1987percolation, rodriguez2013phase} that for each $d\ge 3$ there is an $h_*(d)\in\mathbb{R}$ such that percolation occurs if $h<h_*(d)$ and does not occur if $h>h_*(d)$. Moreover, $h_*(d)>0$ in any dimension $d\ge 3$; for high dimensions this was shown in \cite{rodriguez2013phase} and was strengthened to every $d\ge 3$ in the very recent \cite{Drewitz2017}.

Theorem~\ref{thm:reflection_principle} seems far from the state-of-the-art of these studies but does provide an alternative approach to one of the basic, simple, results in this direction. It was shown in \cite{bricmont1987percolation} that for any strictly convex potential $U$ and any $\varepsilon>0$, in the infinite-volume limit on $\mathbb{Z}^d$ of the random surface measure with Dirichlet boundary conditions, the set of vertices $x$ with $\varphi_x\ge -\varepsilon$ percolates almost surely. We discuss this result in the context of Theorem~\ref{thm:reflection_principle}. Let $G=(V,E)$ be a finite connected graph, let $V_0\subseteq V$ be non-empty, let $U$ be a potential satisfying the monotonicity
condition \eqref{eq:monotonicity_condition_surfaces} and the
assumption \eqref{eq:partition_function_condition} that
$\mu_{U,G,V_0}$ is well-defined and let $\varphi$ be randomly sampled
from $\mu_{U,G,V_{0}}$. Then, for any $\varepsilon>0$ and $v\in V\setminus V_0$, by \eqref{eq:barrier_ineq},
\begin{equation}\label{eq:excursion-set_percolation1}
  \mathbb{P}(V_0\xleftrightarrow{\varphi> -\varepsilon} v) = \mathbb{P}(V_0\xleftrightarrow{\varphi< \varepsilon} v) = 1 - \mathbb{P}(V_0\centernot{\xleftrightarrow{\varphi< \varepsilon}} v) \ge \mathbb{P}(|\varphi_v|< \varepsilon).
\end{equation}
If, in addition, $U$ satisfies the finite-support condition
\eqref{eq:Lipschitz_condition_surfaces} then, relying now on \eqref{eq:Lipschitz_barrier_ineq}, we obtain an inequality in the opposite direction,
\begin{equation}\label{eq:excursion-set_percolation2}
  \mathbb{P}(V_0\xleftrightarrow{\varphi> -\varepsilon} v) \le \mathbb{P}(\varphi_v\in (\varepsilon, \varepsilon+1)).
\end{equation}
When the potential $U$ is strictly convex, one has available the Brascamp-Lieb inequality \cite{brascamp1976extensions} which bounds the variance of $\varphi_x$ by the variance of the Gaussian free field on the same graph. Together with \eqref{eq:excursion-set_percolation1} this can be used to show that the probability that in a discrete cube graph $\{-L,\ldots, L\}^d$ in $\mathbb{Z}^d$, $d\ge 3$, the origin is connected to the boundary of the cube via vertices $v$ with $\varphi_v\ge-\varepsilon$ is uniformly positive as $L$ tends to infinity. Conversely, the same probability in dimension $d=2$ necessarily tends to zero as $L$ increases when $U$ satisfies the finite-support condition and is twice-continuously differentiable on its support, by \eqref{eq:excursion-set_percolation2} and the delocalization results of \cite{milos2015delocalization}.

{\bf Correlation of gradients.} One approach to bounding the fluctuations of random surfaces proceeds via control of the correlations of gradients of the surface. Consider a random surface $\varphi$ sampled from the measure $\mu_{U, G, V_0}$ (see \eqref{eq:random_surface_measure}) with $G=(V,E)$ a finite, connected graph and $V_0\subseteq V$ a non-empty set on which $\varphi$ is set to zero. One may then express the height of the surface at some vertex $x\notin V_0$ as a linear combination of the gradients of the surface. That is, if $\dvec{E}$ denotes the set of oriented edges of $G$ (both orientations of each edge appear in $\dvec{E}$) one writes
\begin{equation*}
  \varphi_{x} = \sum_{(u,v)\in\dvec{E}} c_{(u,v)}(\varphi_u - \varphi_v)
\end{equation*}
for a suitable choice of coefficients $c_{(u,v)}$. Among the many possible choices of these coefficients, Brascamp, Lieb and Lebowitz \cite[Section VII]{brascamp1975statistical} consider the one obtained from the Green's function $g$ of the graph $G$ by writing
\begin{equation}\label{eq:value_at_x_via_Green's function}
  \varphi_{x} = \frac{1}{2}\sum_{(u,v)\in\dvec{E}} (g_u - g_v)(\varphi_u - \varphi_v)
\end{equation}
where $g_y$ is the expected number of visits to $x$ of a simple random walk on $G$ started at $y$ and stopped when it first hits $V_0$ (the factor $\frac{1}{2}$ is needed since each edge is taken with both orientations). The equality \eqref{eq:value_at_x_via_Green's function} is a consequence of the discrete Green's identity $\sum_{\{u,v\}\in E} (a_u - a_v)(b_u-b_v) = -\sum_{u\in V} a_u(\Delta b)_u$, valid for any two functions $a,b$ on $G$ (with $(\Delta b)_u := \sum_{v\in V\,\colon\,\{u,v\}\in E} (b_v - b_u)$), and the fact that $(\Delta g)_y = -\delta_{x,y}$. The identity \eqref{eq:value_at_x_via_Green's function} shows that
\begin{equation}\label{eq:variance_via_gradient_covariances}
  \operatorname{Var}(\varphi_x) = \sum_{\{u,v\}\in E} (g_u - g_v)^2 \operatorname{Var}(\varphi_u - \varphi_v) + \frac{1}{4}\sum_{\substack{(u,v), (z,w)\in\dvec{E}\\\{u,v\}\neq\{z,w\}}} (g_u - g_v)(g_z - g_w)\Cov(\varphi_u - \varphi_v, \varphi_z - \varphi_w)
\end{equation}
and thus highlights how bounding the covariances $\Cov(\varphi_u - \varphi_v, \varphi_z - \varphi_w)$ implies an upper bound on the fluctuations of the surface. Indeed, as pointed out in \cite{brascamp1975statistical}, when $G$ is a discrete cube $\{-L, \ldots, L\}^d$ in the lattice $\mathbb{Z}^d$, $d\ge 3$, and $V_0$ is the boundary of this cube, a decay of the covariances faster than $\|u-z\|_2^{-(2+\varepsilon)}$ for large $\|u-z\|$ would imply that the fluctuations of the surface at the origin remain bounded uniformly in $L$. Such decay is expected but seems quite difficult to establish. Recently, Conlon and Fahim \cite{conlon2015long} used PDE tools to establish asymptotic formulas for the covariance which imply such decay when the potential function $U$ satisfies $0<\inf U''(x) \le \sup U''(x)<\infty$ and certain additional assumptions.

The above discussion provides motivation for studying the gradient-gradient covariances of random surfaces as appear in \eqref{eq:variance_via_gradient_covariances}. It is interesting whether the cluster algorithms used in this work (see Section~\ref{sec:reflection_transformations_cluster_algorithms}) can provide additional tools for controlling such covariances. Indeed, one may try to
reflect the value of $\varphi$ at $v$ around the height
$\varphi_u$, thus reversing the gradient on the edge $(u, v)$. The contribution to the covariance on the event that this reflection leaves the gradient on the edge $(z,w)$ unchanged is exactly zero as the reflection is a measure preserving one-to-one mapping. Thus, this approach connects the problem of estimating the gradient-gradient covariances to the problem of controlling properties of the reflected cluster (of $v$, when reflecting around the height $\varphi_u$) in the cluster algorithm. Related connections in the spin system settings were found by Chayes \cite{chayes1998discontinuity} and Campbell-Chayes \cite{campbell1998isotropic}. It is unclear whether this connection can simplify the problem.

\medskip
{\bf Acknowledgements.} We thank Martin Tassy for helping to clarify the connection between Sheffield's cluster-swapping algorithm and the general cluster algorithms discussed in Section~\ref{sec:reflection_transformations_cluster_algorithms}. We also thank Nishant Chandgotia and Yinon Spinka for several useful discussions.

\bibliographystyle{amsplain}
\bibliography{all}
\end{document}